\documentclass[a4paper]{article}[11pts]

\usepackage[english]{babel}
\usepackage[utf8x]{inputenc}
\usepackage[T1]{fontenc}

\normalsize

\usepackage[a4paper,top=1in,bottom=1in,left=1in,right=1in,marginparwidth=1in]{geometry}

\usepackage{setspace}
\usepackage{amsmath}
\usepackage{amssymb} 
\usepackage{amsthm}
\usepackage{amsfonts, mathtools}
\usepackage{algorithm,algpseudocode}
\usepackage{bm}
\usepackage{comment}
\usepackage{graphicx}
\usepackage{multirow, booktabs}
\usepackage{caption}
\usepackage{subcaption}
\usepackage[colorinlistoftodos]{todonotes}
\usepackage[colorlinks=true, allcolors=blue]{hyperref}

\newcommand{\p}{\bm{p}}
\newcommand{\Bb}{\bm{b}}
\newcommand{\Dd}{\bm{d}}
\newcommand{\Aa}{\bm{a}}

\newcommand{\E}{\mathbb{E}}
\newcommand{\R}{\mathbb{R}}
\newcommand{\prob}{\mathbb{P}}
\newcommand{\Z}{\mathcal{Z}}

\DeclareMathOperator*{\argmin}{arg\,min}

\DeclarePairedDelimiter\floor{\lfloor}{\rfloor}

\newtheorem{proposition}{Proposition}
\newtheorem{lemma}{Lemma}
\newtheorem{example}{Example}
\newtheorem{assumption}{Assumption}
\newtheorem{theorem}{Theorem}

\newtheorem{definition}{Definition}

\usepackage{natbib}

\title{Simple and Fast Algorithm for Binary Integer and Online Linear Programming}
\author{Xiaocheng Li$^\dagger$ \and Chunlin Sun$^\ddagger$ \and  Yinyu Ye$^\dagger$}
\date{\small 
$^\dagger$Department of Management Science and Engineering, Stanford University\\
$^\ddagger$ Institute for Computational and Mathematical Engineering, Stanford University\\
$\{$chengli1, chunlin, yyye$\}$@stanford.edu\\
}

\begin{document}
\maketitle

\doublespacing

\begin{abstract}
In this paper, we develop a simple and fast online algorithm for solving a class of binary integer linear programs (LPs) arisen in general resource allocation problem. The algorithm requires only one single pass through the input data and is free of doing any matrix inversion. It can be viewed as both an approximate algorithm for solving binary integer LPs and a fast algorithm for solving online LP problems. The algorithm is inspired by an equivalent form of the dual problem of the relaxed LP and it essentially performs (one-pass) projected stochastic subgradient descent in the dual space. We analyze the algorithm in two different models, stochastic input and random permutation, with minimal technical assumptions on the input data. The algorithm achieves $O\left(m \sqrt{n}\right)$ expected regret under the stochastic input model and $O\left((m+\log n)\sqrt{n}\right)$ expected regret under the random permutation model, and it achieves $O(m \sqrt{n})$ expected constraint violation under both models, where $n$ is the number of decision variables and $m$ is the number of constraints. The algorithm enjoys the same performance guarantee when generalized to a multi-dimensional LP setting which covers a wider range of applications. In addition, we employ the notion of permutational Rademacher complexity and derive regret bounds for two earlier online LP algorithms for comparison. Both algorithms improve the regret bound with a factor 
of $\sqrt{m}$ by paying more computational cost. Furthermore, we demonstrate how to convert the possibly infeasible solution to a feasible one through a randomized procedure. Numerical experiments illustrate the general applicability and effectiveness of the algorithms.

\end{abstract}

\section{Introduction}

In this paper, we present simple and fast online algorithms to approximately solve a general class of binary integer (online) linear programs (LP). Specifically, we consider binary integer LPs that take the following form
\begin{align}
 \tag{ILP}  \max \ \ & \bm{r}^\top \bm{x}  \label{eqn:ILP} \\
    \text{s.t. }\ & \bm{A} \bm{x} \le \bm{b} \nonumber  \\ 
    & x_j \in \{0,1\}, \ \ j=1,...,n, \nonumber
\end{align}
where $\bm{r} = (r_1,...,r_n) \in \R^n$, $\bm{A} = (\bm{a}_1,...,\bm{a}_n) \in \R^{m\times n},$ and $\bm{b} = (b_1,...,b_m) \in \R^m.$ The decision variables are $\bm{x}=(x_1,...,x_n)\in \{0,1\}^n$. Different specifications of the above formulation cover a wide range of classic problems and modern applications: secretary problem \citep{ferguson1989solved}, knapsack problem \citep{kellerer2003knapsack}, resource allocation problem \citep{vanderbei2015linear}, generalized assignment problem \citep{conforti2014integer}, network routing problem \citep{buchbinder2009online}, matching problem \citep{mehta2005adwords}, etc. 

Our algorithm is a primal-dual algorithm inspired by an equivalent form of the dual problem of the LP relaxation. The key is to perform one pass of projected stochastic subgradient descent in the dual space and to decide the primal solution based on the dual iterate solution in an online fashion. The algorithm requires only one single pass through the input ($\bm{r}$ and $\bm{A}$) of an problem instance, and is free of doing any matrix inversion. When the right-hand-side $\bm{b}$ scales linearly with $n$, we show that the algorithm outputs a solution with an expected optimality gap of $O(m\sqrt{n})$ and constraint violation of $O(m\sqrt{n}),$ under minor technical assumptions (on $\bm{r}$ and $\bm{A}$).

From the perspective of \textbf{integer LP}, our algorithm is an efficient approximate algorithm that features of provable performance guarantee. In general, integer LP is NP-hard, and its LP relaxation technique has been widely used in designing integer LP algorithm. Our algorithm is also inspired by the relaxed LP, but it directly outputs an integer solution to the relaxed LP. The solution can thus be viewed as an approximate solution to both the integer LP and the relaxed LP.

From the perspective of \textbf{online LP}, to the best of our knowledge, our algorithm is the most simple and fastest online LP algorithm so far. Furthermore, the algorithm analysis is conducted under the two prevalent models: the stochastic input model and the random permutation model. The stochastic input model assumes that the columns of the LP together with the corresponding coefficients in the objective function are drawn i.i.d. from an unknown distribution, but our assumption is weaker than the previous literature in that the strong convexity is not assumed for the underlying stochastic program. The random permutation model assumes that the columns together with the coefficients are drawn in a randomly permuted order, and it better captures the possibly non-stationary, heterogeneous and/or adversarial input data of LPs. Under the latter model, our assumption is weaker than all previous works in that we allow negative input data of the underlying LP.

\subsection{Related Literature}

The algorithms developed in this paper can be viewed as a randomized algorithm to solve large-scale (integer) LPs. The literature on large-scale LP algorithms traced back to the early works on column generation algorithm \citep{ford1958suggested, dantzig1963linear}. In recent years, statistical structures underlying the input of LP have been taken into consideration. Sampling-based/randomized LP algorithms are derived to handle large number of constraints in the LP of Markov Decision Processes 
\citep{de2004constraint, lakshminarayanan2017linearly}, the standard form of LP \citep{vu2018random}, robust convex optimization \citep{calafiore2005uncertain}, etc.
\cite{de2004constraint} studied an approximate LP problem arising from the approximate dynamic programming approach and developed a sampling scheme to reduce the number of constraints under certain statistical assumptions. A subsequent work \citep{lakshminarayanan2017linearly} developed a soft approach by replacing the original LP constraints by a smaller set of constraints that are constructed from positive linear combinations of the original ones. \cite{vu2018random} discussed the standard LP formulation and introduced a random projection method to approximately solve large-scale LP in the light of Johnson-Lindenstrauss Lemma. Compared to this line of works, our algorithm utilizes the dual LP and is free of solving any small-scale or reduced-size LP. Our algorithm can be viewed as an online and efficient version of the dual projected subgradient (DPG) algorithm for LP \citep{beck2017first}. Our algorithm exploits only one column at a time without replacement for stochastic subgradient descent in each iteration, whereas the general dual project subgradient algorithm requires the whole constraint matrix and conducts matrix multiplication in each iteration. 
Recently, another stream of works studied first-order algorithms, mainly alternating direction method of multipliers (ADMM) method, for solving large-scale LPs \citep{yen2015sparse, wang2017new, lin2018admm}. Compared to the algorithms developed therein, our algorithm requires one single pass through the inputs of the LP and does not involve any optimization sub-routine nor matrix inversion.   

Our algorithms and analyses also contribute to the literature of online linear programming. The formulation studied in this paper is the same as the previous works (See \citep{molinaro2013geometry, agrawal2014dynamic, kesselheim2014primal, gupta2014experts, li2019online} among others). Among all these algorithms, the algorithm proposed in this paper is the \textbf{simplest and fastest} one. The online LP literature mainly considered two models -- the stochastic input model and the random permutation model. These two models put different assumptions on the coefficients in the constraint matrix and in the objective function. In Section \ref{SIM} and Section \ref{RPM}, we analyze our algorithm under these two models respectively. Under the stochastic input model, our assumption on the distribution is minimal than the previous works including \citep{li2019online} and other specific applications such as the network revenue management problem (See \citep{jasin2013analysis, jasin2015performance, bumpensanti2018re} among others).
Compared to the literature under random permutation (See \citep{molinaro2013geometry, agrawal2014dynamic, kesselheim2014primal, gupta2014experts} among others), we allow the inputs of the LP to take negative values and consider the regime of large right-hand-side. Specifically, the previous works investigated the necessary and sufficient conditions on the right-hand-side of the LP $\bm{b}$ and the number of constraint $m$ for the existence of an $\epsilon$-competitive (near-optimal) online LP algorithm. Alternatively, we research the question, when the right-hand-side $\bm{b}$ grows linearly with the number of decision variables $n$, whether the algorithm could achieve a better performance than $\epsilon$-competitiveness. More importantly, our work is the first analysis for online LP algorithms under the random permutation model that allows negative data values for the input of the LP which have wider applications such as double-auction markets.

Similar or special forms of the online LP problem have been extensively studied in different application contexts. These problems include secretary problem \citep{kleinberg2005multiple, arlotto2019uniformly}, auction problem \citep{zhou2008budget, balseiro2019learning}, network revenue management problem \citep{jasin2013analysis, jasin2015performance, bumpensanti2018re}, and resource allocation problem \citep{asadpour2019online, jiang2019online, lu2020dual}. The common point for all these problems is that there is an underlying LP and the coefficients of the LP are specified by the corresponding application context. Consequently, their algorithm designs and analyses rely on the structure of the underlying LP, such as all-one constraint matrix (in secretary problem), binary constraint matrix (in resource allocation problem), finite support of the random coefficients (in network revenue management problem), etc. 
In this work, we study online LP problem in its general form and thus the design and analyses could provide potential theoretical and algorithmic insights for these specific applications. 

Furthermore, we employ the notion of permutational Rademacher complexity and develop a machinery for more general regret analysis under random permutation model. This new machinery enables analyses of two previously proposed algorithms \citep{agrawal2014dynamic, kesselheim2014primal} under a regime with $\bm{b}/n \rightarrow \bm{d}$, and with the relaxation of the previous assumption on non-negativeness of LP coefficients. In comparison with our fast algorithm, the two algorithms \citep{agrawal2014dynamic, kesselheim2014primal} are more computationally costly, but they reduce the regret and the constraint violation with an order of $\sqrt{m}.$ This analytical framework is of independent interest and could be applied for the online algorithm analysis under the random permutation model in a more general context.

The problem of online LP can seemingly be viewed as a special form of online convex optimization with constraints (\texttt{OCOwC}). However, these two problems are studied separately in the literature. Our paper establishes a connection between these two lines of research by identifying the dual form of online LP problem as a special form of the primal problem of \texttt{OCOwC}. The literature on \texttt{OCOwC} \citep{mahdavi2012trading, yu2017online, yuan2018online} that employs stochastic gradient descent methods thus can be applied to analyze the dual objective/solution for the online LP problem under stochastic input model. Our contribution here is to identify the connections between the primal and dual objectives, and the constraint violation of the primal online LP problem. Moreover, an important distinction between the online LP problem and the \texttt{OCOwC} problem is that when computing the regret, the former considers a stronger benchmark where the decision variables are allowed to take different values at each time period, while the later considers a stationary benchmark where the the decision variables are required to be the same at each time period. 
Another elegant paper \citep{agrawal2014fast} developed and analyzed fast algorithms for the problem of online convex programming. It differs from the online LP problem in that the formulation therein considered a simpler form of the constraint which requires the average of the decision variables chosen throughout the process belongs to a convex set. It thus corresponds to a setting in an online LP problem where the constraint matrix is an all-one matrix. 

Another stream of literature studied the random reshuffling method for stochastic gradient descent (SGD) algorithm in minimizing a finite sum of convex component functions \citep{gurbuzbalaban2015random, ying2018convergence, safran2019good}. Their study of random reshuffling method is mainly focused on the question whether SGD will converge faster under sampling with or without replacement. The method shares a similar spirit with our algorithm under the random permutation model, but the focus of our problem is the presence of the constraints. Also, the differentiation between the stochastic input model and the random permutation model in our paper emphasizes more on the generation mechanism for the inputs for the LP, whereas their study of random reshuffling method concerns more about the better sampling scheme for SGD given the fixed data.

\section{Integer Linear Program and Main Algorithm}

\subsection{Integer LP, Primal LP, and Dual LP}
Consider the LP relaxation of the integer LP (\ref{eqn:ILP})
\begin{align}
  \tag{P-LP}  \max \ \ & \bm{r}^\top \bm{x} \label{eqn:P-LP}  \\
    \text{s.t. }\ & \bm{A} \bm{x} \le \bm{b} \nonumber  \\ 
    & \bm{0} \le \bm{x} \le \bm{1}. \nonumber
\end{align}
The dual problem of (\ref{eqn:P-LP}) is 
\begin{align}
  \tag{D-LP}  \min \ \ & \bm{b}^\top \bm{p} + \bm{1}^\top \bm{s} \label{eqn:D-LP}  \\
    \text{s.t. }\ &  \bm{A}^\top \bm{p} + \bm{s} \ge \bm{r} \nonumber  \\ 
    & \bm{p} \ge \bm{0}, \bm{s}\ge \bm{0}, \nonumber
\end{align}
where the decision variables are $\bm{p}\in \R^m$ and $\bm{s} \in \R^n.$ Throughout this paper, $\bm{0}$ and $\bm{1}$ denote all-zero and all-one vector, respectively. We will use ILP, P-LP, and D-LP to refer to both the optimization problem and their optimal objective values. Evidently, we have the follow relation between the optimal objective values, $$\text{ILP} \le \text{P-LP} = \text{D-LP}.$$
This natural relation provides the foundation for the wide usage of LP relaxation in solving integer linear programs \citep{conforti2014integer}. Now, we start from the linear programs P-LP and D-LP to derive a simple algorithm for the ILP problem, by utilizing an underlying structure of the LPs.

We denote the optimal solutions to (\ref{eqn:P-LP}) and (\ref{eqn:D-LP}) with $\bm{x}^*$, $\bm{p}^*_n$, and $\bm{s}^*$, and the optimal solutions to (\ref{eqn:ILP}) as $\bar{\bm{x}}^*.$ From the complementary condition, we know that 
\begin{equation} \label{dualOpt}
x_j^* = \begin{cases}
1,& r_j > \bm{a}_j^\top \bm{p}^*_n  \\
0,& r_j < \bm{a}_j^\top \bm{p}^*_n 
\end{cases}
\end{equation}
for $j=1,...,n.$ When $r_j = \bm{a}_j^\top \bm{p}^*_n$, the optimal solution $x_j^*$ may be a non-integer value. The implication of this optimality condition is that the primal optimal solution $\bm{x}^*$ can be largely determined by the dual optimal solution $\bm{p}^*_n.$ To derive our algorithm, we first introduce an informal statistical assumption on the input of the LPs, and we will further elaborate the assumption in the later sections. 

\begin{assumption}{(Informal).} We assume the column-coefficient pair $(r_j, \bm{a}_j)$'s are i.i.d. sampled from unknown distribution $\mathcal{P}$.
\label{informal}
\end{assumption}

If we denote the right-hand-side $\bm{b} = n \bm{d},$ as noted by \cite{li2019online}, an equivalent form the dual problem that only involves decision variables $\bm{p}$ can be obtained from (\ref{eqn:D-LP}) by plugging the constraints into the objective and removing the decision variables $\bm{s}$. Specifically, consider
\begin{align}
\tag{SAA} \min_{\bm{p}}\ & f_n(\bm{p})  = \bm{d}^\top \bm{p} + \frac{1}{n} \sum_{j=1}^n \left(r_j-\bm{a}_j^\top \bm{p}\right)^+  \label{SAA} \\
\text{s.t. }\ & \bm{p}\ge \bm{0}. \nonumber
\end{align}
where $(\cdot)^+$ denotes the positive part function. Under Assumption \ref{informal}, all the terms in the summation in (\ref{SAA}) are independent with each other. Thus, the function $f_n(\bm{p})$ can be viewed as a \textit{sample average approximation} (SAA) of the stochastic program
\begin{align}
\tag{SP} \min_{\bm{p}}\ & f(\bm{p})  = \bm{d}^\top \bm{p} + \E_{(r,\bm{a})\sim\mathcal{P}}\left[\left(r-\bm{a}^\top \bm{p}\right)^+\right]  \label{SP} \\
\text{s.t. }\ & \bm{p}\ge \bm{0}. \nonumber
\end{align}

Denote the optimal solution to (\ref{SP}) as $\bm{p}^*$. Then the optimal dual solution $\bm{p}_n^*$ to $f_n(\bm{p})$ (equivalently, the original dual program D-LP) can be viewed as an approximate to $\bm{p}^*$. We refer to the previous work \citep{li2019online} for an extensive discussion on the convergence analysis of $\bm{p}_n^*$ to $\bm{p}^*$. 

\subsection{Main Algorithm}
Now, we present the main algorithm -- Simple Online Algorithm. First, it is an online algorithm that observes the inputs of the LP sequentially and decides the value of decision variable $x_t$ immediately after each observation $(r_t, \bm{a}_t)$. Second, the algorithm is a dual-based algorithm. It maintains a dual vector $\p_t$ and determines $x_t$ in a similar way as the optimality condition (\ref{dualOpt}). At each time $t$, it updates the vector with the new observation $(r_t, \bm{a}_t)$ and projects to the non-negative orthant to ensure the dual feasibility.

\begin{algorithm}[ht!]
\caption{Simple Online Algorithm}
\label{alg:SOA}
\begin{algorithmic}[1]
\State Input: $\bm{d}=\bm{b}/n$
\State Initialize $\bm{p}_1 = \bm{0}$ 
\For {$t=1,..., n$}
\State Set 
$$x_t = \begin{cases}
1,& r_t >\bm{a}_t^\top \bm{p}_t \\
0,& r_t \le \bm{a}_t^\top \bm{p}_t 
\end{cases}$$
\State Compute
\begin{align*}
    \bm{p}_{t+1} & = \bm{p}_t + \gamma_t \left(\bm{a}_tx_t - \bm{d}\right) \\
    \bm{p}_{t+1} & = \bm{p}_{t+1} \vee \bm{0}
\end{align*}
\EndFor
\State Output: $\bm{x} = (x_1,...,x_n)$
\end{algorithmic}
\end{algorithm}

The key of the algorithm is the updating formula for $\bm{p}_t$, namely Step 5 in Algorithm \ref{alg:SOA}. For two vectors $\bm{u},\bm{v} \in \R^m$, $\bm{u} \vee \bm{v} = \left(\max\{u_1,v_1\},...,\max\{u_m,v_m\}\right)^\top$ denotes the elementwise maximum operator.  Specifically, the update from $\bm{p}_t$ to $\bm{p}_{t+1}$ can be interpreted as a \textit{projected stochastic subgradient descent} method for optimizing the problem (\ref{SAA}). Concretely, the subgradient of the $t$-th term in (\ref{SAA}) evaluated at $\bm{p}_t$,
\begin{align*}
    \partial_{\bm{p}} \left(\bm{d}^\top \bm{p} + \left(r_t- \bm{a}_t^\top \bm{p}\right)^+\right)\Bigg|_{\p = \p_t} & = \bm{d} - \bm{a}_t I(r_t>\bm{a}_t^\top\bm{p}) \Big|_{\p = \p_t} \\
    & = \bm{d}-\bm{a}_t x_t
\end{align*}
where the second line is due to the specification of $x_t$ as the step 4 in the Algorithm \ref{alg:SOA}. Throughout this paper,  $I(\cdot)$ denotes the indicator function.  The dual updating rule indeed implements the stochastic subgradient descent in the dual space. We defer the rigorous analysis of the algorithm performance and the choice of the step size $\gamma_t$ to later sections. 

As for the computational aspect, Algorithm \ref{alg:SOA} requires only one pass through the data and is free of matrix multiplications. Generally, algorithms use LP relaxation to progressively solve integer LPs. In certain sense, the solution given by the optimal solution to the relaxed LP (\ref{eqn:P-LP}) can be viewed as a non-integer approximation to the optimal solution of the according integer LP (\ref{eqn:ILP}). In contrast, the integer solution output from Algorithm \ref{alg:SOA}, though most likely not the optimal solution to the integer LP (\ref{eqn:ILP}), can be viewed as an integer approximation to the (non-integer) optimal solution of the LP (\ref{eqn:P-LP}). Consequently, Algorithm \ref{alg:SOA} works as an approximate algorithm to solve the integer LP (\ref{eqn:ILP}), and it is inspired by but not directly utilizing the corresponding LP (\ref{eqn:P-LP}).
 
\subsection{Performance Measures}

We analyze the algorithm in two aspects -- optimality gap (regret) and constraint violation. The optimality gap measures the difference in objective values for the algorithm output and the true optimal solution. Since Algorithm \ref{alg:SOA} does not ensure a feasible solution, we need to account the total amount of constraint violations for its output. In this paper, we focus on this bi-objective performance measure for two reasons. First, there may be ways to transform an infeasible solution to a feasible solution which absorbs the constraint violation into the regret (as Theorem 2 in \cite{li2019online}), but it may require stronger assumptions on the inputs of the (integer) LP. In this paper, we aim to develop theoretical results under minimal assumptions on the input. In this light, it might be challenging to combine the two objectives into one. In Section \ref{feasibleAlg}, we elaborate more on this aspect and discuss a variant of Algorithm \ref{alg:SOA} that guarantees feasibility. Second, the bi-objective performance measure is aligned with the literature on the online convex optimization with constraints (\texttt{OCOwC}); the same objective is considered in \citep{mahdavi2012trading, yu2017online, yuan2018online}. Additionally, as we will see in the later sections, there is a natural connection between the primal optimality gap, dual optimality gap, and the constraint violation. 

In the following two sections, we will formalize the assumptions and analyze the algorithm in two different settings.

\section{Stochastic Input Model}

\label{SIM}

In this section, we formalize and analyze the algorithm under the statistical assumption proposed in the last section. Concretely, we discuss the performance of Algorithm \ref{alg:SOA} when the inputs of an (integer) LP follow the stochastic input model which assumes the column-coefficient pair $(r_j,\bm{a}_j)$'s are i.i.d. generated. LPs and integer LPs that satisfy this model naturally arise from application contexts where each pair represents a customer/order/request. In particular, at each time $t$, $\bm{a}_t$ can be interpreted as a customer request for the resources while $r_t$ represents the revenue that the decision maker receives from accepting this request. The binary decision variable $x_t$ represents the decision of acceptance or rejection of the $t$-th request. In such context, the dual vector $\bm{p}_t$ conveys a meaning of dual price and it assigns a value $\bm{a}_t^\top \bm{p}_t$ to the $t$-th request. In Algorithm \ref{alg:SOA}, the dual-based decision rule will accept this request if the revenue received $r_t$ exceeds its assigned value. We recently learned that \cite{lu2020dual} also produced results for online optimization problem with a similar algorithm under the i.i.d. model where the random vector $(r_j, \bm{a}_j)$ has a finite support. 

\subsection{Assumptions and Performance Measures}

\label{performanceMeasure}

The following assumption formalizes the statistical assumption on $(r_j, \bm{a}_j)$ in an i.i.d. setting.
\begin{assumption}[Stochastic Input]
\label{stoch}
We assume
\begin{itemize}
    \item[(a)] The column-coefficient pair $(r_j,\bm{a}_j)$'s are i.i.d. sampled from an unknown distribution $\mathcal{P}.$
    \item[(b)] There exist constants $\bar{r}$ and $\bar{a}$ such that 
    $|r_j|\le \bar{r}$ and $\|\bm{a}_{j}\|_\infty\le \bar{a}$ for $j=1,...,n.$
    \item[(c)] The right-hand-side $\bm{b}=n\bm{d}$ and there exist positive constants $\underline{d}$ and $\bar{d}$ such that $\underline{d}  \le d_i\le \bar{d}$ for $i=1,...,m.$
\end{itemize}
\end{assumption}

We emphasize that the constants $\bar{r}$, $\bar{a}$, $\underline{d}$ and $\bar{d}$ only serve for analysis purpose and are assumed unknown a priori. Also, we allow dependence between components in $(r_j, \bm{a}_j)$'s. Besides the boundedness, we have put minimal assumption on $r_j$ and $\bm{a}_j$. This is different from the previous work \citep{li2019online} where stronger assumptions are introduced to ensure a strong convexity for the stochastic program $f(\bm{p})$ (\ref{SP}). As a result, the convergence of $\bm{p}_t$ can be established under the assumptions here, as least not with the same convergence rate as \citep{li2019online}. For part (c), the assumption on right-hand-side side provides a service level guarantee, i.e., it ensures a fixed proportional of customers/orders can be fulfilled as the total number of customers (market size) $n$ increases. We use $\Xi$ to denote the family of distributions that satisfy Assumption \ref{stoch} (b).

Next, we formally define the regret and the constraint violation. Denote the optimal objective values of the ILP and P-LP as $Q_n^*$ and $R_n^*$, respectively. The objective value obtained by the algorithm output is
$$R_n = \sum_{j=1}^n r_j x_j.$$ 
The quantity of interest is the optimality gap $Q_n^* - R_n$, which has an upper bound 
$$Q_n^* - R_n \le R_n^* - R_n.$$
The expected optimality gap is
$$\Delta_{n}^{\mathcal{P}} = E\left[R_n^* - R_n\right]$$
where the expectation is taken with respect to $(r_j, \bm{a}_j)$'s. Define regret as the worst-case optimality gap
$$\Delta_{n} = \sup_{\mathcal{P} \in \Xi} \Delta_{n}^{\mathcal{P}}.$$
Thus the regret bound derived in this paper has a two-fold meaning: (i) an upper bound for the optimality gap of solving the integer LP; (ii) a regret bound for the online LP problem. Provided that we do not assume any knowledge of the distribution $\mathcal{P}$, this type of distribution-free bound is legitimate. We emphasize that the definition of regret for the online LP problem differs from that for the online convex optimization problem \citep{hazan2016introduction} where the decision variables for the offline optimal are restricted to take the same value over time; in contrast, we allow $x_1,...,x_n$ to take different values in defining $R_n^*$. 

Another performance measure for Algorithm \ref{alg:SOA} is the \textit{constraint violation}, 
$$ v(\bm{x}) =  \|\left(\bm{A}\bm{x}-\bm{b}\right)^+ \|_2$$
where $\bm{A}$ is the constraint coefficient matrix, $\bm{b}$ is the right-hand-side constraint, and $\bm{x}$ is the solution. 
We aim to quantify the expected norm of the constraint violation. Similar to the regret, we seek for an upper bound for the constraint violation that is not dependent on the distribution $\mathcal{P}.$

\subsection{Algorithm Analyses}
First, we analyze the dual price sequence $\bm{p}_t$'s. The following lemma states that the dual price $\bm{p}_t$'s under Algorithm \ref{alg:SOA} will remain bounded throughout the process, and this is true with probability $1$. 

\begin{lemma}
\label{iidBound} Under Assumption \ref{stoch}, if the step size $\gamma_t \le 1$ for $t=1,...,n$ in Algorithm \ref{alg:SOA}, then
$$\|\bm{p}^*\|_2\leq\frac{\bar{r}}{\underline{d}},$$  
$$\|\bm{p}_t\|_2 \leq{\frac{2\bar{r}+m(\bar{a}+\bar{d})^2}{\underline{d}}} + m(\bar{a}+\bar{d}).$$ with probability $1$ for $t=1,...,n$, where $\p_t$'s are specified by Algorithm \ref{alg:SOA}.
\end{lemma}

\begin{proof}
See Section \ref{PFiidBound}.
\end{proof}

Essentially, this boundedness property arises from the updating formula. The intuition is that if the dual price $\bm{p}_t$ becomes large, then most of the ``buying'' requests (with $\bm{a}_j$ being positive) will not be rejected, and this will lead to a decrease of the dual price when computing $\bm{p}_{t+1}$. As we will see later, the norm of $\bm{p}_t$ appears frequently in the algorithm performance analysis, in term of both the regret and the constraint violation. Therefore the implicit boundedness of $\bm{p}_t$ becomes important in that it saves us from having to do explicit projection, on both computational and modeling level. On one hand, projecting $\bm{p}_t$ into a compact set at every step might be computational costly; on the other hand, this compact set requires more prior knowledge on underlying LP.

\begin{theorem}
\label{theoiid}
Under Assumption \ref{stoch}, if the step size $\gamma_t=\frac{1}{\sqrt{n}}$ for $t=1,...,n,$ the regret and 
expected constraint violation of Algorithm \ref{alg:SOA} satisfy
$$\E[R_n^* - R_n] \le m(\bar{a}+\bar{d})^2\sqrt{n}$$
$$
  \E\left[v(\bm{x})\right]\le
\left({\frac{2\bar{r}+m(\bar{a}+\bar{d})^2}{\underline{d}}}+m(\bar{a}+\bar{d})\right)\sqrt{n}. 
$$
hold for all $m, n\in \mathbb{N}^+$ and distribution $\mathcal{P}\in \Xi.$
\end{theorem}

\begin{proof}
See Section \ref{PFtheoiid}.
\end{proof}

The number of constraints $m$ decides the dimension of the dual price vectors $\bm{p}_t$'s. Both the regret and the expected constraint violation is $O(m\sqrt{n})$.  Algorithm \ref{alg:SOA} conducts subgradient descent updates in the dual space but the performance is measured by the primal objective. The key idea for the proof of Theorem \ref{theoiid} is to establish the connections between primal objective, dual objective, and constraints violation through the lens of the updating formula for $\bm{p}_t$. The proof mimics the classic analysis for convex online optimization problems \citep{hazan2016introduction}. It provides an explanation for why the seemingly related problems of online LP and online convex optimization with constraints (\texttt{OCOwC}) are studied separately in the literature. On one hand, the online LP literature has been focused on studying the primal objective value as the performance measure. On the other hand, the \texttt{OCOwC} problem \citep{mahdavi2012trading, yu2017online, yuan2018online} also studied mainly the primal objective under online stochastic subgradient descent algorithms. However, it is the dual problem of online LP that corresponds to a special form of the primal problem in the \texttt{OCOwC} literature. Our contribution is to identify this correspondence and to establish the primal-dual connection for online LP problem when applying stochastic subgradient descent.  

\section{Random Permutation Model}

\label{RPM}

In this section, we consider a random permutation model where the column-coefficient pair $(r_j, \bm{a}_j)$ arrives in a random order. The values of $(r_j, \bm{a}_j)$'s can be chosen adversarially at the start. However, the arrival order of $(r_j, \bm{a}_j)$'s is uniformly distributed over all the permutations. This characterizes a weaker condition than the previous stochastic input model and the analysis under this model allows more general application of the algorithm. There are two ways to interpret Algorithm \ref{alg:SOA} under this random permutation model. First, it can be interpreted as an online algorithm that solves an online LP problem under data generation assumptions that are weaker than the i.i.d. assumptions discussed in the last section. Hence, the stochastic input model can be viewed as a special case of the random permutation model. In particular, the latter captures the case when there exists possibly non-stationarity or adversary for the inputs of the LPs. Second, from the perspective of solving integer LPs, the permutation creates the randomness for integer LPs when there is no inherent randomness with the coefficients. As we will see, this artificially created randomness is sufficient for Algorithm \ref{alg:SOA} to provide provable performance guarantee comparable to the case of the stochastic input model.

\begin{example}
Consider a multi-secretary problem
\begin{align*}
    \max\  & \sum_{j=1}^{n} r_jx_j  \\
 \text{s.t. }\ & \sum_{j=1}^n x_j \le b
\end{align*}
with $b\in\mathbb{N}^+$ and $n=2b$. Moreover, $r_1=...=r_b = 1$ and $r_{b+1}=...=r_n=2.$
\end{example}

This example of multi-secretary problem illustrates the idea and necessity of doing random permutation. This problem in its original form does not satisfy the i.i.d. assumption, and if one solves the problem in its original order, there is no way we can infer about the ``good'' candidates $\{r_j\}_{j=b+1}^n$ in the later half by just observing the first half of the data $\{r_j\}_{j=1}^b$. However, if we randomly permute the $r_j$'s, then the problem becomes
\begin{align*}
    \max\  & \sum_{j=1}^{n} r_{\sigma(j)}x_{\sigma(j)}  \\
    \text{s.t. }\ & \sum_{j=1}^n x_{\sigma(j)} \le b
\end{align*}
where $(\sigma(1),...,\sigma(n))$ is a random permutation of $(1,...,n).$ Intuitively, for this new problem, it is very likely that we obtain a good knowledge of the whole data $\{r_j\}_{j=1}^n$ by simply observing the first few samples. Generally speaking, the random permutation technique handles this type of problem where there is no inherent randomness. In this section, we analyze the regret and the constraint violation of Algorithm \ref{alg:SOA} under the random permutation model. Later in Section \ref{permutRC}, we provide a more systematic treatment of the random permutation model and analyze the performance of two previously proposed algorithms.

\subsection{Assumption and Performance Measures}

In parallel to the stochastic input model, we formalize the random permutation model as follows.
\begin{assumption}[Random Permutation] We assume
\begin{itemize}
    \item[(a)] The column-coefficient pair $(r_j, \bm{a}_j)$ arrives in a random order. 
\item[(b)] There exist constants $\bar{r}$ and $\bar{a}$ such that 
    $|r_j|\le \bar{r}$ and $\|\bm{a}_{j}\|_\infty\le \bar{a}$ for $j=1,...,n.$
    \item[(c)] The right-hand-side $\bm{b}=n\bm{d}$ and there exists positive constant $\underline{d}$ and $\bar{d}$ such that $\underline{d}  \le d_i\le \bar{d}$ for $i=1,...,m.$
\end{itemize}
\label{permut}
\end{assumption}
Assumption \ref{permut} part (b) and (c) are identical to the stochastic input model. Denote the input data set $\mathcal{D}=\{(r_j, \Aa_j): 1\le j\le n\}$. Part (a) in Assumption \ref{permut} states that we observe a permuted realization of the data set. Additionally, we make the following assumption on the data set $\mathcal{D}.$
\begin{assumption}
The problem inputs are in a general position, namely for any price vector $\bm{p}$, there
are at most $m$ columns such that $\bm{a}_j^\top \bm{p} = r_j.$
\label{generalPosi}
\end{assumption}

This assumption is not necessarily true for all the data set $\mathcal{D}$. However, as pointed out by \citep{devanur2009adwords}, one can always randomly perturb $r_t$'s by arbitrarily small amount. In this way, the assumption will be satisfied, and the effect of this perturbation on the objective can be made arbitrarily small. Define
\begin{equation}
    x_j(\bm{p}) = \begin{cases}
1, & r_j > \bm{a}_j^\top \bm{p}, \\
0, & r_j \le \bm{a}_j^\top \bm{p}
\end{cases} 
\label{dualThres}
\end{equation}
and $\bm{x}(\bm{p}) = (x_1(\bm{p}),...,x_n(\bm{p})).$
Lemma \ref{gp} tells that if $\bm{p}_n^*$ is used in (\ref{dualThres}), the corresponding primal solution should be feasible and close to the primal optimal solution. The complementarity condition (\ref{dualOpt}) does not imply anything about the primal optimal solution when $r_j=\bm{a}_j^\top \bm{p}_n^*.$ The thresholding rule (\ref{dualThres}), as it appears in Algorithm \ref{alg:SOA}, takes a conservative standpoint by setting $x_t=0$ if $r_j=\bm{a}_j^\top \bm{p}$ when we use the dual price $\bm{p}.$ Essentially, the general position in Assumption \ref{generalPosi} ensures that $r_j=\bm{a}_j^\top \bm{p}$ will happen at most $m$ times for any $\bm{p}$ and Lemma \ref{gp} justifies that the effect of being conservative on these points with the optimal dual price $\bm{p}^*_n$ is marginal.

\begin{lemma}
$x_j(\bm{p}^*_n)\leq x_j^*$ for all $j=1,...,n$ and under Assumption \ref{generalPosi}, $x_j(\bm{p}^*_n)$ and $x_j^*$ differs for no more than $m$ values of $j$. It implies that, under Assumption \ref{generalPosi}, if one uses the optimal dual solution $\bm{p}_n^*$ in the thresholding rule, the obtained solution will no greater than the primal optimal solution and they will be different for at most $m$ entries. 
\label{gp}
\end{lemma}

\begin{proof}
See Lemma 1 in \citep{agrawal2014dynamic}.
\end{proof}

As for the performance measure, we use the same notations as in Section \ref{performanceMeasure}. 
The expected optimality gap
$$\delta_{n}^{\mathcal{D}} = R_n^* - \E\left[R_n\right].$$
Throughout this section, the expectation is always taken with respect to a random permutation on the data set $\mathcal{D}$, unless otherwise stated. Given the data set $\mathcal{D},$ $R_n^*$ is a deterministic quantity, so it is unnecessary to take an expectation for it. This also underscores the difference between the stochastic input model and the random permutation model. That is, the randomness arises from the data (the LP input) in the stochastic input model, whereas it arises from the ordering of the data in the random permutation model. Define regret as the worst-case optimality gap
$$\delta_{n} = \sup_{\mathcal{D} \in \Xi_D} \delta_{n}^{\mathcal{D}}$$
where $\Xi_D$ denotes all the data sets that satisfy Assumption \ref{permut} (b) and Assumption \ref{generalPosi}. In this way, the regret quantifies the worst-case performance of the algorithm for all possible inputs data $\mathcal{D}$. 

\subsection{Algorithm Analyses}

First, the following lemma states that the boundedness property of the dual price remains the same as in the stochastic input model. Its proof is identical to the stochastic input model, since the proof of Lemma \ref{iidBound} only relies on the boundedness assumption on $(r_j, \bm{a}_j)$'s but not the statistical assumption about the data generation.

\begin{lemma}
\label{permutBound} Under Assumption \ref{permut} and Assumption \ref{generalPosi}, we have
$$\|\bm{p}^*_n\|_2\leq\frac{\bar{r}}{\underline{d}},$$  
$$\|\bm{p}_t\|_2 \leq{\frac{2\bar{r}+m(\bar{a}+\bar{d})^2}{\underline{d}}} + m(\bar{a}+\bar{d}).$$ with probability $1$ for all $t$, where $\p_t$'s are specified by Algorithm \ref{alg:SOA}.
\end{lemma}

To facilitate our derivation, we define a scaled version of the primal LP (\ref{eqn:P-LP}),
\begin{align} \label{eqn:S-LP}
  \tag{$s$-S-LP}  \max \ \ & \sum_{j=1}^s r_jx_j   \\
    \text{s.t. }\ & \sum_{j=1}^s a_{ij}x_j \le \frac{sb_i}{n}  \nonumber \\
    & 0 \le x_j \le 1\ \text{ for } j=1,...,s.\nonumber
\end{align}
for $s=1,...,n$. Denote its optimal objective value as $R_s^*$. The following proposition relates $R_s^*$ with $R_n^*.$ 

\begin{proposition}
For $s>\max\{16\bar{a}^2,e^{16\bar{a}^2},e\},$ the following inequality holds 
    \label{partialLP}
    \begin{equation}\label{R_k_R_n}
        \frac{1}{s}\mathbb{E} \left[R_{s}^*\right]
    \geq
    \frac{1}{n} R_n^*
    -
    \frac{m\bar{r}}{n}-\frac{\bar{r}\log s}{\underline{d}\sqrt{s}}
    -
    \frac{m\bar{r}}{s}.
    \end{equation}
    for all $s\le n\in\mathbb{N}^+$ and $\mathcal{D}\in \Xi_{D}.$
    \label{importantLemma}
\end{proposition}

\begin{proof}
See Section \ref{PFimportant}.
\end{proof}

Intuitively, in the random permutation model, the observations $\{(r_j, \bm{a}_j)\}_{j=1}^s$ collected until time $s$ are less informative to infer the future observations than the case of the stochastic input model. However, Proposition \ref{importantLemma} tells that the scaled LP (\ref{eqn:S-LP}) constructed based on the first $s$ observations will achieve a similar expected optimal objective value (after scaling) compared with the original problem with all $n$ observations. Note that $\E[R_s^*]/s = \E[R_n^*]/n$ is evidently true in the stochastic input model, where the expectation is taken with respect to the distribution $\mathcal{P}.$ The additional terms on the right-hand-side of (\ref{R_k_R_n}) captures the information toll (on the order of $\log s/\sqrt{s}$) for the assumption relaxation from the stochastic input model to the random permutation model. This proposition bridges the gap between past and future observations in the random permutation model, i.e., what one can tell about the future samples based on the past observations. Comparatively, this gap between past and future observations is taken care by the sampling from same distribution $\mathcal{P}$ in the stochastic input model. 

The regret analysis in Theorem \ref{theoPermut} builds on Proposition \ref{importantLemma}. The idea is that if $\bm{p}_{s+1}$ from Algorithm \ref{alg:SOA} is a reasonably good dual solution to the scaled LP
($s$-S-LP), and plus that $\E[R_s^*]/s \approx R_n^*/n,$ $\bm{p}_{s+1}$ should also be a good dual solution for the rest of inputs, and specifically for the upcoming sample $(r_{s+1}, \bm{a}_{s+1}).$

\begin{theorem}
\label{theoPermut}
  Under Assumption \ref{permut} and \ref{generalPosi}, if the step size $\gamma_t=\frac{1}{\sqrt{n}}$ for $t=1,...,n,$ the regret and expected constraint violation of Algorithm \ref{alg:SOA} satisfy
$$
        R_n^*-\mathbb{E}[R_n] \le O\left((m+\log n)\sqrt{n}\right)
$$
$$
  \E\left[v(\bm{x})\right]\le
O(m\sqrt{n}). 
$$
for all $m, n \in \mathbb{N}^+$ and $\mathcal{D} \in \Xi_{D}.$
\end{theorem}

\begin{proof}
See Section \ref{PFtheoPermut}.
\end{proof}

Compared to the stochastic input model, the regret upper bound under random permutation model contains an extra term of $O(\sqrt{n}\log n)$, while the constraint violation in two models enjoys the same upper bound. Note that Proposition \ref{importantLemma} and Theorem \ref{theoPermut} do not require the non-negativeness assumption of the LP input. As far as we know, this is the first online LP analysis under random permutation model without the non-negativeness assumption.


\section{Multi-dimensional Integer Linear Program}

In this section, we discuss a multi-dimensional extension of (\ref{eqn:ILP})
\begin{align}
   \tag{Multi-ILP} \max \ \ & \sum_{j=1}^n \bm{r}^\top_j \bm{x}_j \label{eqn:multiILP}  \\
    \text{s.t. }\ & \sum_{j=1}^n \bm{A}_j \bm{x}_j \le \bm{b} \nonumber  \\ 
    & \bm{1}^\top \bm{x}_j \le 1, \ \ \bm{x}_j \in \{0,1\}^k, \ \ j=1,...,n \nonumber
\end{align}
where $\bm{r}_j = (r_{j1},...,r_{jk}) \in \R^k$, $\bm{A}_j = (\bm{a}_{j1},...,\bm{a}_{jk}) \in \R^{m\times k}$, and $\bm{a}_{jl} = (a_{1jl},...,a_{mjl})^\top,$ for $j=1,...,n$ and $l=1,...,k.$ The decision variables are $\bm{x}=\left(\bm{x}_1,...,\bm{x}_n\right)$ where $\bm{x}_j=(x_{j1},...,x_{jk})^\top$ for $j=1,...,n$. The right-hand-side capacity $\bm{b}=(b_1,...,b_m)$ is the same as the one-dimensional setting (\ref{eqn:ILP}). The formulation is called as multi-dimensional because the binary decision variable $x_j$ in (\ref{eqn:ILP}) is replaced with a vector $\bm{x}_j\in \{0,1\}^k.$ It covers a wider range of applications than the previous setting, including adwords problem \citep{mehta2005adwords}, generalized assignment problem \citep{conforti2014integer}, resource allocation problem \citep{asadpour2019online}, etc.

Algorithm \ref{alg:multi-D} is a natural generalization of Algorithm \ref{alg:SOA} in the multi-dimensional setting. The idea is to maintain a dual price as Algorithm \ref{alg:SOA}, and then to use the dual price to identify the most profitable dimension for each order. The decision of $\bm{x}_t$ (Step 7 in Algorithm \ref{alg:multi-D}) arises from the complementarity condition of (\ref{eqn:multiILP}). Accordingly, Assumption \ref{multi-D} generalizes the stochastic input and random permutation assumptions in the previous sections.  

\begin{algorithm}[ht!]
\caption{Simple Online Algorithm for Multi-dimensional ILP}
\label{alg:multi-D}
\begin{algorithmic}[1]
\State Input: $d$
\State Initialize $\bm{p}_1 = \bm{0}$ 
\For {$t=1,...,n$}
\State Set $v_t=\max\limits_{l=1,...,k} r_{tl}-\bm{a}_{tl}^\top\bm{p}_t$
\If {$v_t>0$}
\State Pick an index $l_t$ randomly from the non-empty set
$$\left\{l: r_{tl}-\bm{a}_{tl}^\top\bm{p}_t = v_t \right\}$$
\State Set
$$x_{tl}=\begin{cases}
1,& l=l_t \\
0,& \text{otherwise} 
\end{cases}
$$
\Else
\State Set $\bm{x}_t=\bm{0}$
\EndIf
\State Compute
\begin{align*}
    \bm{p}_{t+1} & = \bm{p}_t + \frac{1}{\sqrt{n}}\left(\bm{A}_t\bm{x}_t - \bm{d}\right) \\
    \bm{p}_{t+1} & = \bm{p}_{t+1} \vee \bm{0}
\end{align*}
\EndFor
\State Output: $\bm{x} = (\bm{x}_1,...,\bm{x}_n)$
\end{algorithmic}
\end{algorithm}

\begin{assumption}
We assume
\begin{itemize}
    \item[(a)] (Stochastic Input). The column-coefficient pair $(\bm{r}_j,\bm{A}_j)$'s are i.i.d. sampled from an unknown distribution $\mathcal{P}.$
    \item[(a')] (Random Permutation). The column-coefficient pair $(\bm{r}_j, \bm{A}_j)$ arrives in a random order. The problem is in a general position; $\bm{x}(\bm{p}^*_n)$ and $\bm{x}^*$ differs for no more than $m$ values of $t$.
    \item[(b)] There exist constants $\bar{r}$ and $\bar{a}$ such that 
    $|\bm{r}_j|\le \bar{r}$ and $\|\bm{A}_{j}\|_\infty\le \bar{a}$ for $j=1,...,n.$ 
    \item[(c)] The right-hand-side $\bm{b}=n\bm{d}$ and there exist positive constants $\underline{d}$ and $\bar{d}$ such that $\underline{d}  \le d_i\le \bar{d}$ for $i=1,...,m.$
\end{itemize}
\label{multi-D}
\end{assumption}

\begin{theorem}
\label{theoMD}
Under the stochastic input and random permutation model in Assumption \ref{multi-D}, the regret and constraint violation of Algorithm \ref{alg:multi-D} are the same as Theorem \ref{theoiid} and Theorem \ref{theoPermut}, respectively.
\end{theorem}

Theorem \ref{theoMD} states the regret and constraint violation of Algorithm \ref{alg:multi-D} are the same as the previous one-dimensional setting and in particular, not dependent on the dimension $k$ of $\bm{x}_t$'s.

\section{More Regret Analysis under Random Permutation Model via Permutational Rademacher Complexity}

\label{permutRC}

In this section, we analyze the regret of two ``slower'' algorithms \citep{agrawal2014dynamic, kesselheim2014primal} of online LP under the random permutation model. Since they all involved solving scaled LPs, they are slower than the algorithm proposed in this paper. Both \citep{agrawal2014dynamic} and \citep{kesselheim2014primal} analyzed the algorithms under the random permutation model and the right-hand-side assumption that $\bm{b}/n\rightarrow 0$, and provided constant competitiveness ratio guarantee. Instead, we consider the regime where $\bm{b}/n\rightarrow \bm{d}>0$ and provide sublinear regret upper bounds. Furthermore, both previous works assume the entries $a_{ij}$'s in the constraint matrix to be non-negative but for here we remove this assumption. In this sense, the analyses in this section complements to the previous results. Recall that in Proposition \ref{importantLemma}, we establish the connection between the optimal solutions of the scaled LP and the original LP. Now we extend the result and connect history and future observations (under the random permutation model) in a more systematic way; for example, if a dual vector $\bm{p}$ performs well in history observations, it should perform roughly as well in the future observations. Such connection can be established easily under random input model because the past and future observations are drawn from the same distribution. Things become trickier for the random permutation model because there is no 
restriction on where the underlying data $\mathcal{D}$ comes from. 

In this section, we first present the \textit{Permutational Rademacher Complexity} that quantifies the gap between history and future observations, and then discuss the regret and constraint violation for two previous algorithms. Specifically, at time $t$, if we compute the dual price $\p$ based on the history input $\{(r_j,\Aa_j)\}_{j=1}^{t}$ and set $x_j=I\left(r_j>\Aa_j^\top \p\right)$, then intuitively, we should have the objectives (and the constraints) in the past and future roughly match after scaling properly,
$$\frac{1}{t}\E \left[\sum_{j=1}^{t} r_j I(r_j>\Aa_j^\top \p)\right] \approx \frac{1}{n-t}\E \left[\sum_{j=t+1}^{n} r_j I(r_j>\Aa_j^\top \p)\right]$$
$$\frac{1}{t}\E \left[\sum_{j=1}^{t} \Aa_j I(r_j>\Aa_j^\top \p)\right] \approx \frac{1}{n-t}\E \left[\sum_{j=t+1}^{n} \Aa_j I(r_j>\Aa_j^\top \p)\right]$$
when $(r_j,\Aa_j)$'s arrive in a random permutation order from the dataset $\mathcal{D}$ and the expectation is taken with respect to the permutation. To establish such connection, the idea is to view both hand sides of the above as two random functions of $\p$ and to analyze the evaluation of $\p$ on these functions. The Permutational Rademacher Complexity formalizes the idea and paves the way for performance analysis under random permutation model.

\subsection{Permutational Rademacher Complexity}

Consider set $\mathcal{Z}_n=\{z_1,...,z_n\}$ with $z_j \in \R^k, j=1,...,n$ and a family of functions $\mathcal{F} = \{f: \R^k \rightarrow \R\}$ (to be specified later). Throughout this section, we use the subscript to indicate the cardinality of a set. For function $f\in \mathcal{F}$ and $\mathcal{S}\subset \Z$, denote $$\bar{f}(\mathcal{S}) = \frac{1}{|\mathcal{S}|}\sum_{x\in\mathcal{S}} f(x)$$
as the average function value on the set $\mathcal{S}.$ The definition of the permutational Rademacher complexity and its relation with conditional Rademacher complexity largely mimic the analyses of the transductive learning problem in \citep{tolstikhin2015permutational}. 

\begin{definition}[Permutational Rademacher Complexity and Conditional Rademacher Complexity (See Definition 3 in \citep{tolstikhin2015permutational})] For any $1\le s\le t-1$, permutational Rademacher complexity (PRC) is defined as follows:
$$
Q_{t,s}(\mathcal{F},\Z_t) =
\mathbb{E}
\sup\limits_{f\in\mathcal{F}}
\left|
\bar{f}(\Z_s)
-
\bar{f}(\tilde{\Z}_l)
\right|,
$$
where $\Z_s$ is subset of $\Z_t$ with $s$ elements sampled uniformly
without replacement and $\tilde{\Z}_l = \Z_t \backslash \Z_s$, $l = t-s$. The expectation is taken with respect to the random sampling of $\Z_s$.

Conditional Rademacher complexity (CRC) is defined as follows:
$$
R_{t}(\mathcal{F}, \Z_t) =
\mathbb{E}
\sup\limits_{f\in\mathcal{F}}
\left|\frac{1}{t}
\sum\limits_{j=1}^{t} \epsilon_j f(z_j)
\right|,
$$
where $\Z_t = \{z_1,...,z_t\}$ and $\epsilon_j$'s are i.i.d. random variables following Rademacher distribution ($P(\epsilon_j=1) = P(\epsilon_j=-1) = 1/2$). The expectation is taken with respect to $\epsilon_j$'s.
\end{definition}

Both the above two quantities are dependent on the set $\Z_t$ because for the PRC, the two subsets $\Z_s$ and $\tilde{\Z}_l$ are sampled from $\Z_t$ and for CRC, it is computed based on the function values of the elements in $\Z_t.$ Both PRC and CRC are
deterministic with a given function class $\mathcal{F}$ and conditional on $\Z_t$. However, they could be random variables if the set $\Z_t$ is random, for example, a randomly sampled subset of $\Z_n.$

The following lemma explains the motivation for the definition of permutational Rademacher complexity and it is inspired from Theorem 2 in \citep{tolstikhin2015permutational}. Specifically, at time $t$, if the function $f$ is evaluated on the random set $\Z_t$ and its complement $\tilde{Z}_{t'}$, the difference between the average evaluations can be upper bounded by the expected PRC. In the context of online LP, if we want to measure the gap between the objectives (a family functions specified by dual price $\p$), a convenient way would be to compute the corresponding PRC.

\begin{lemma}

$\Z_t$ is a subset of $\Z_n$ obtained by uniform sampling without replacement, and $\tilde{\Z}_{t'}=\Z_n\backslash \Z_t,$ $t'=n-t.$ Without loss of generality, assume $t\ge {t'}$, then the following inequality holds for all $s<{t'},$
$$
\mathbb{E} \left[\sup\limits_{f\in\mathcal{F}}\left|
\bar{f}(\Z_t) - \bar{f}(\tilde{\Z}_{t'})
\right|\Big|\Z_n\right]
\leq
\mathbb{E}\left[
Q_{t,s}(\mathcal{F},\Z_{t})\Big| \Z_n \right]
$$
where the expectation is taken with respect to the random sampling of $\Z_t$ from $\Z_n$
\label{lemma:bound_diff_between_test_and_train}
\end{lemma}

Let
$\mathcal{F}_{\bm{p}} 
=
\left\{ 
f_{\bm{p}}: f_{\bm{p}}(r,\bm{a}) = rI\left( r>\bm{a}^T\bm{p}  \right)\right\}$ denote a family of functions $f:\R^{m+1}\rightarrow \R$ indexed by the parameter $\p$, and let $(r_1,\Aa_1),....,(r_n,\Aa_n)$ be a random permutation of the dataset $\mathcal{D}.$ Denote $\Z_t = \{(r_1,\Aa_1),...,(r_t,\Aa_t)\}$, and then $\Z_t$ can be viewed as a randomly sampled subset of $\mathcal{D}.$ The following lemma relates the PRC with the classic notion of CRC, and the benefit is that the CRC  under random permutation model possesses a natural upper bound. The proofs for Lemma \ref{lemma:bound_diff_between_test_and_train} and Lemma \ref{lemma:PRC} are deferred to Section \ref{ProofPRC}.

\begin{lemma}
\label{lemma:PRC}
The following inequalities hold for PRC and CRC of the family $\mathcal{F}_{\p}$ and dataset $\mathcal{D}$ that satisfies Assumption \ref{permut},
$$
\left|
Q_{t,\floor{t/2}}(\mathcal{F}_{\bm{p}},\Z_t)
-
R_{t}(\mathcal{F}_{\bm{p}},\Z_t)
\right|
\leq
\frac{4\bar{r}}{\sqrt{t}},
$$
$$
R_{t}(\mathcal{F},\mathcal{Z}_t)
\leq
\frac{\sqrt{2\bar{r}^2m\log n}}{\sqrt{t}}.
$$
\end{lemma}

Returning to the LP problem, the following proposition answers the question proposed at the beginning of this section. It tells that if we employ the same dual price on two subsets of size $t$ and $n-t$, the difference in objective values and constraints consumption is roughly on the order of $\sqrt{\frac{m\log n}{ \min\{t, n-t\}}}.$ 

\begin{proposition}
If $\{(r_j,\Aa_j)\}_{j=1}^n$ is a random permutation of dataset $\mathcal{D}$ and satisfies Assumption \ref{permut}, we have 
$$\E \left[ \sup_{\p\ge0} \left\vert \frac{1}{n-t}\sum_{j=t+1}^n r_{j}I(r_j>\Aa_j^\top \p)  - 
\frac{1}{t}\sum_{j=1}^t r_{j}I(r_j>\Aa_j^\top \p)  \right\vert\right]  \le  \frac{4\bar{r}}{\sqrt{\min\{t,n-t\}}}+ \frac{2\sqrt{2\bar{r}^2m\log n}}{\sqrt{\min\{t,n-t\}}}$$
$$\E  \left[ \sup_{\p\ge0} \left\vert \frac{1}{n-t} \sum_{j=t+1}^n a_{ij}I(r_j>\Aa_j^\top \p)-
\frac{1}{t}\sum_{j=1}^t a_{ij}I(r_j>\Aa_j^\top \p)  \right\vert \right]\le \frac{4\bar{a}}{\sqrt{\min\{t,n-t\}}} +\frac{2\sqrt{2\bar{a}^2m\log n}}{\sqrt{\min\{t,n-t\}}} $$
for any $1\le i \le m$, $1\le t\le n-1$.
\label{prop_permut_rdm}
\end{proposition}

\subsection{Performance Analyses for Two ``Slower'' Algorithms}

Algorithm \ref{alg:DLA} was first proposed in \citep{agrawal2014dynamic} and then refined in \citep{li2019online}. The idea is to construct a dual price $\p_t$ at each time $t$ based on the observed sample and to use this dual price for the decision of time $t+1.$ The algorithm is much slower than Algorithm \ref{alg:SOA} since at each iteration, an LP (of growing size) is solved to compute the dual price. For the analysis, the PRC theory presented earlier thus provides a machinery to relate the evaluation of $\p_t$ on the past $t$ samples with that of the incoming sample at time $t+1$. \cite{agrawal2014dynamic} analyzed Algorithm \ref{alg:DLA} under the random permutation model and with the regime $\bm{b}/n\rightarrow 0,$ whereas \cite{li2019online} analyzed the algorithm under the stochastic input model and with the regime $\bm{b}/n\rightarrow \bm{d}.$ Theorem \ref{theorem_DLA} complements to these discussions for the random permutation model with the regime $\bm{b}/n\rightarrow \bm{d},$ and its proof is referred to Section \ref{ProofPRC}. 

\begin{algorithm}[ht!]
\caption{Dynamic Learning Algorithm \citep{agrawal2014dynamic}} \label{alg:DLA}
\begin{algorithmic}[1]
\State Input: $\bm{d}$
\State Let $\p_1 = \bm{0}$
\For {$t=1,2,...,n$}
\State $${x}_{t} = \begin{cases} 1, & \text{ if } r_{t} > \bm{a}_{t}^\top \bm{p}_t\\
0, & \text{ if } r_{t} \le \bm{a}_{t}^\top \bm{p}_t
\end{cases}$$
\State The scaled primal LP is
\begin{align*}
\max\ & \sum_{j=1}^{t} r_j x_j 
    \\ 
    \text{s.t.}\ & \sum_{j=1}^{t} a_{ij}x_j \le td_i ,\ \  i=1,...,m \\
    & 0 \le x_{j} \le 1, \ \  j=1,...,{t}
\end{align*}
\State Solve its dual problem and obtain the optimal dual variable $\bm{p}_{t+1}$
\begin{align*}
  \bm{p}_{t+1} &=  \argmin_{\bm{p} \ge 0} \ \ \sum_{i=1}^m d_ip_i + \frac{1}{t}\sum_{j=1}^{t} \left(r_j-\sum_{i=1}^m a_{ij}p_i\right)^+ 
\end{align*}
\EndFor
\end{algorithmic}
\end{algorithm}

\begin{theorem}
\label{theorem_DLA}
Under Assumption \ref{permut} and \ref{generalPosi}, the regret and expected constraint violation of Algorithm \ref{alg:DLA} satisfy
$$R_n^*-\mathbb{E}[R_n] \le O\left(\sqrt{mn}\right)$$
$$\E[\bm{A}\bm{x}-\bm{b}]\le O(\sqrt{mn}\log n)$$
for all $m, n \in \mathbb{N}^+$ and $\mathcal{D} \in \Xi_{D}.$ Here $\bm{x}$ is the output of Algorithm \ref{alg:DLA}.
\end{theorem}

Theorem \ref{theorem_DLA} shows that the regret and constraint violation can be reduced by a factor of $O(\sqrt{m})$ compared with Algorithm \ref{alg:SOA}, with the price of computation cost. This order is the same as the lower and upper bound for online convex optimization \citep{hazan2016introduction} up to a logarithm factor. 

The primal-beats-dual algorithm (Algorithm \ref{alg:PBD}) was proposed in \citep{kesselheim2014primal} and it can be viewed as a primal version of Algorithm \ref{alg:DLA}. At each time $t$, it solves the primal scaled LP and projects the fractional solution $\tilde{x}_{t}^{(t)}$ to a binary value. Therefore it involves solving an LP at each time period and is slower than Algorithm \ref{alg:SOA}. The analysis of objective value builds upon the Proposition \ref{importantLemma} and the analysis of the constraint violation employs a backward Martingale argument.

\begin{algorithm}[ht!]
\caption{Primal-beats-dual Algorithm \citep{kesselheim2014primal}} \label{alg:PBD}
\begin{algorithmic}[1]
\State Input: $\bm{d}$
\State Let $\p_1 = \bm{0}$
\For {$t=1,2,...,n$}
\State Solve the scaled LP 
\begin{align*}
\max\ & \sum_{j=1}^{t} r_j x_j 
    \\ 
    \text{s.t.}\ & \sum_{j=1}^{t} a_{ij}x_j \le td_i ,\ \  i=1,...,m \\
    & 0 \le x_{j} \le 1, \ \  j=1,...,{t}
\end{align*}
\State Denote the optimal solution as $\tilde{\bm{x}}^{(t)}=(\tilde{x}_1^{(t)},...,\tilde{x}_t^{(t)})$ 
\State $${x}_{t} = \begin{cases} 1, & \text{ with probability } \tilde{{x}}^{(t)}_t\\
0, & \text{ with probability } 1- \tilde{{x}}^{(t)}_t
\end{cases}$$
\EndFor
\end{algorithmic}
\end{algorithm}

\begin{theorem}
Under Assumption \ref{permut} and \ref{generalPosi}, the regret and expected constraint violation of Algorithm \ref{alg:PBD} satisfy
$$R_n^*-\mathbb{E}[R_n] \le O\left(\sqrt{mn}\right)$$
$$\E[v(\bm{x})]\le O(\sqrt{mn}\log n)$$
for all $m, n \in \mathbb{N}^+$ and $\mathcal{D} \in \Xi_{D}.$ Here $\bm{x}$ is the output of Algorithm \ref{alg:PBD}.
\label{theorem_PBD}
\end{theorem}

Theorem \ref{theorem_PBD} states that the regret and constraint violation of Algorithm \ref{alg:PBD} are on the order of $O(\sqrt{mn})$. Like Algorithm \ref{alg:DLA}, the extra computation cost here also helps improve the algorithm performance in terms of $m.$ Its proof is deferred to Section \ref{proofTheoPBD}.

\section{Algorithm Discussion}

\subsection{Obtaining Feasible Solution}

We present a simple approach to convert the solution obtained from Algorithm \ref{alg:SOA} into a feasible solution. Let $\bm{x} = (x_1,...,x_n)$ be a solution by Algorithm \ref{alg:SOA}, and $S_+ = \{t:x_t=1, t=1,...,n\}$ be the index set of nonzero $x_t$'s
and $n_+=|S_+|$  be the cardinality of $S_+.$ The idea is to randomly select a subset of $S_+$ and force $x_t=0$ for indices in this subset. Note that the expected total constraint violation is sublinear in $n,$ we only need to select a small proportion of $x_t$'s and force them to be $0.$ Specifically, define the maximum constraint violation quantity over all constraints:
\begin{equation*}
    v = \frac{1}{\sqrt{n}\log n}
    \max\limits_{i=1,...,m}\left\{ 
    \left(\sum_{t=1}^n {a}_{it}x_{t} - b_i\right)^+\right\}.
\end{equation*}
Moreover, we require $v\ge1.$ 
We choose a set $S_0\subset S_+$ uniformly with $|S_0|=\min\left\{\left[\frac{2vn_+\log n}{\underline{d}\sqrt{n}}\right]+1,n_+\right\}$, and let 
$$\hat{x}_t = \begin{cases}
 0, & t\in S_0 \\
 x_t, & t\notin S_0
\end{cases}$$
for $t=1,...,n.$ The following theorem characterizes the properties of $\hat{x}_t.$

\begin{theorem}
\label{theorem_Feasibel} 
If $n>\max\left\{16, \underline{d}^2, \left(\frac{6\bar{a}}{\underline{d}}\right)^4\right\}$ and $\sqrt{n}>\frac{12\bar{a}(\bar{r}+(\bar{a}+\underline{d})^2m)\log n}{\underline{d}^2}$, then $\hat{\bm{x}} = (\hat{x}_1,...,\hat{x}_n)$ is a feasible solution with probability at least $1-\frac{2}{n}$. Also, a feasible solution $\bm{x}$ can be constructed based on $\hat{\bm{x}}$ s.t.,
$$\E\left[R_n^*-\bm{r}^\top \bm{x}\right] \le O((m+\log n)\sqrt{n})$$
for all $m,n \in \mathbb{N}_+.$
The results hold under both the stochastic input model and the random permutation model, and the expectation is taken with respect to $\mathcal{P}$ or the random permutation accordingly.
\end{theorem}

Theorem \ref{theorem_Feasibel} tells that in a large-$n$-small-$m$ regime, precisely when $n\ge O(m^2 \log n)$, we can easily obtain a feasible solution with high probability based on the output of Algorithm \ref{alg:SOA} by randomly selecting $O(\sqrt{n}\log n)$ number of $x_t$ and forcing them to be $0$. Furthermore, the newly obtained solution does not change the regret much. The theorem provides a guideline of the implementation of Algorithm \ref{alg:SOA} for the binary LP setting when a feasible solution is desired. 

\subsection{Feasible Online Algorithm}
\label{feasibleAlg}

Algorithm \ref{alg:SFA} is another natural variant of Algorithm \ref{alg:SOA} (or Algorithm \ref{alg:multi-D}) that outputs feasible solutions. Compared with Algorithm \ref{alg:SOA}, Algorithm \ref{alg:SFA} sets $x_t=1$ only when the constraints permit. This design is more aligned with the online LP algorithms that guarantees feasibility.  \cite{li2019online} provided a regret analysis framework for this type of feasible algorithms, and the key is to analyze the stopping time of constraint violation and the remaining resources for binding constraints. In this paper, the assumptions on $(r_j, \bm{a}_j)$ are parsimonious and they might be not sufficient to derive an upper bound on these two key quantities. Numerically, we observe that this feasible algorithm, in comparison with Algorithm \ref{alg:SOA}, does not compromise the performance in terms of the regret. We will elaborate more on its numerical performance in the next section and leave the regret analysis of this algorithm as an open question.

\begin{algorithm}[ht!]
\caption{Simple Feasible Algorithm}
\label{alg:SFA}
\begin{algorithmic}[1]
\State Input: $d$
\State Initialize $\bm{p}_1 = \bm{0}$
\For {$t=1,..., n$}
\State Set 
$$\tilde{x}_t = \begin{cases}
1,& r_t >\bm{a}_t^\top \bm{p}_t \\
0,& r_t \le \bm{a}_t^\top \bm{p}_t 
\end{cases}$$
\State Compute
\begin{align*}
    \bm{p}_{t+1} & = \bm{p}_t + \gamma_t\left(\bm{a}_t\tilde{x}_t - \bm{d}\right) \\
    \bm{p}_{t+1} & = \bm{p}_{t+1} \vee \bm{0}
\end{align*}
\State If constraints permit, set $x_t = \tilde{x}_t;$ otherwise set $x_t=0.$
\EndFor
\State Output: $\bm{x} = (x_1,...,x_n)$
\end{algorithmic}
\end{algorithm}

\subsection{Nonstationary Algorithm}

We consider another variant of the algorithm that takes into account the resource consumption while doing the subgradient descent. The intuition is similar to the action-history-dependent algorithm in \citep{li2019online}.
If excessive resources are consumed in the early periods, the remaining resource $\bm{b}_t$ will shrink, and this nonstationary algorithm will accordingly push up the dual price and be more inclined to reject an order. On the contrary, if we happen to reject a lot orders at the beginning and it results in too much remaining resources, the algorithm will lower down the dual price so as to accept more orders in the future. In numerical experiments, this nonstationary algorithm performs better, but it is still on the same order of regret and constraint violation as Algorithm \ref{alg:SOA}. The open question is if there exists a first-order algorithm that is free of re-solving any linear programs and could achieve $O(\log n)$ regret, possibly under stronger statistical assumptions.

\begin{algorithm}[ht!]
\caption{Simple Nonstationary Algorithm}
\label{alg:SNA}
\begin{algorithmic}[1]
\State Input: $d$
\State Initialize $\bm{p}_1 = \bm{0},$ $\bm{b}_0=\bm{b}$
\For {$t=1,..., n$}
\State Set 
$$x_t = \begin{cases}
1,& r_t >\bm{a}_t^\top \bm{p}_t \\
0,& r_t \le \bm{a}_t^\top \bm{p}_t 
\end{cases}$$
\State Update
$$\bm{b}_t = \bm{b}_{t-1} - \bm{a}_t x_t$$
\State Compute
\begin{align*}
    \bm{p}_{t+1} & = \bm{p}_t + \gamma_t\left(\bm{a}_tx_t - \frac{\Bb_{t}}{n-t}\right) \\
    \bm{p}_{t+1} & = \bm{p}_{t+1} \vee \bm{0}
\end{align*}
\EndFor
\State Output: $\bm{x} = (x_1,...,x_n)$
\end{algorithmic}
\end{algorithm}

\subsection{Step Size}

From the perspective of stochastic gradient descent, Algorithm \ref{alg:SOA} adopts step size $\gamma_t = 1/\sqrt{n}$. With the same analyses, it can be shown that the algorithm also works with step size $\gamma_t = 1/\sqrt{t}$. In the literature of stochastic optimization and online optimization, some algorithms \citep{moulines2011non, lacoste2012simpler} employed a step size of $1/t^\alpha$ for $\alpha \in [1/2,1]$ which commonly requires a strong convexity assumption. For the problem of online LP or integer LP, even if we can enforce a strong convexity assumption, it might be unrealistic to assume the knowledge of the strong convexity constant (necessary for a smaller step size) as a priori. There are also discussions on the usage of averaging method \citep{xiao2010dual, juditsky2014deterministic} or an adaptive approach for choosing the step size \citep{flammarion2015averaging, lei2019adaptivity} for stochastic gradient descent. We leave it as an  open question whether these designs will result in better online LP algorithms.

\section{Numerical Experiments}

The first experiment compares the performance of Algorithm \ref{alg:SOA}, Algorithm \ref{alg:SFA}, and Algorithm \ref{alg:SNA} in terms of regret and constraint violation. Algorithm \ref{alg:SOA} is implemented with two different choices of step size $\gamma_t$. In this experiment, $m=10,$ $a_{ij}$'s and $r_{j}$'s are sampled i.i.d. from Unif$[0,2]$. For each value of $n$, we run $100$ simulation trials and in each trial, $d_i$'s are sampled i.i.d. from Unif$[1/3,2/3].$ The average regret and constraint violation over all the simulation trials are shown in Figure \ref{exp1}. We plotted normalized regret and constraint violation, which is absolute regret and constraint violation divided the optimal objective value and the L$_2$ norm of the constraint capacity, respectively.
We observe that the step size of $1/\sqrt{n}$ results in larger constraint violation but smaller regret compared with the step size of $1/\sqrt{t}$. This is because for the step size of $1/\sqrt{n}$, the updating of the dual vector $\p_t$ is slower. Consequently, more requests will be accepted at early stage and the constraint violation is larger in the end. The non-stationary algorithm (Algorithm \ref{alg:SNA}) performs better than the simple algorithm (Algorithm \ref{alg:SOA}) with $\gamma_t=1/\sqrt{t}$. Also, the feasible algorithm (Algorithm \ref{alg:SFA}) guarantees feasibility, i.e. zero constraint violation; it produces slightly larger regret, but the regret is still on the order of $\sqrt{n}$. 

\begin{figure}[ht!]
\begin{subfigure}{.5\textwidth}
  \centering
  \includegraphics[width=.93\linewidth]{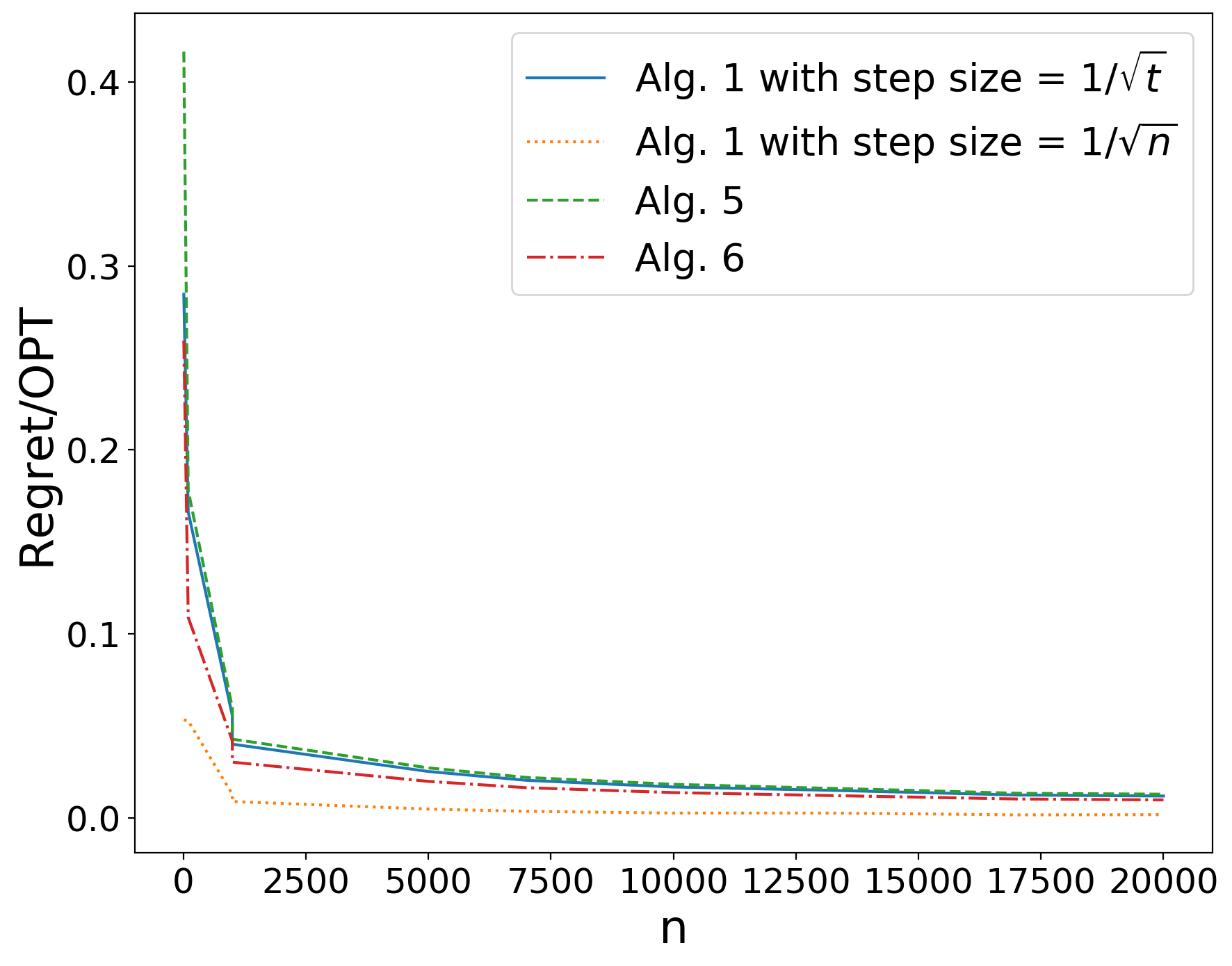}
  \caption{Regret}
  \label{exp1a}
\end{subfigure}%
\begin{subfigure}{.5\textwidth}
  \centering
  \includegraphics[width=.98\linewidth]{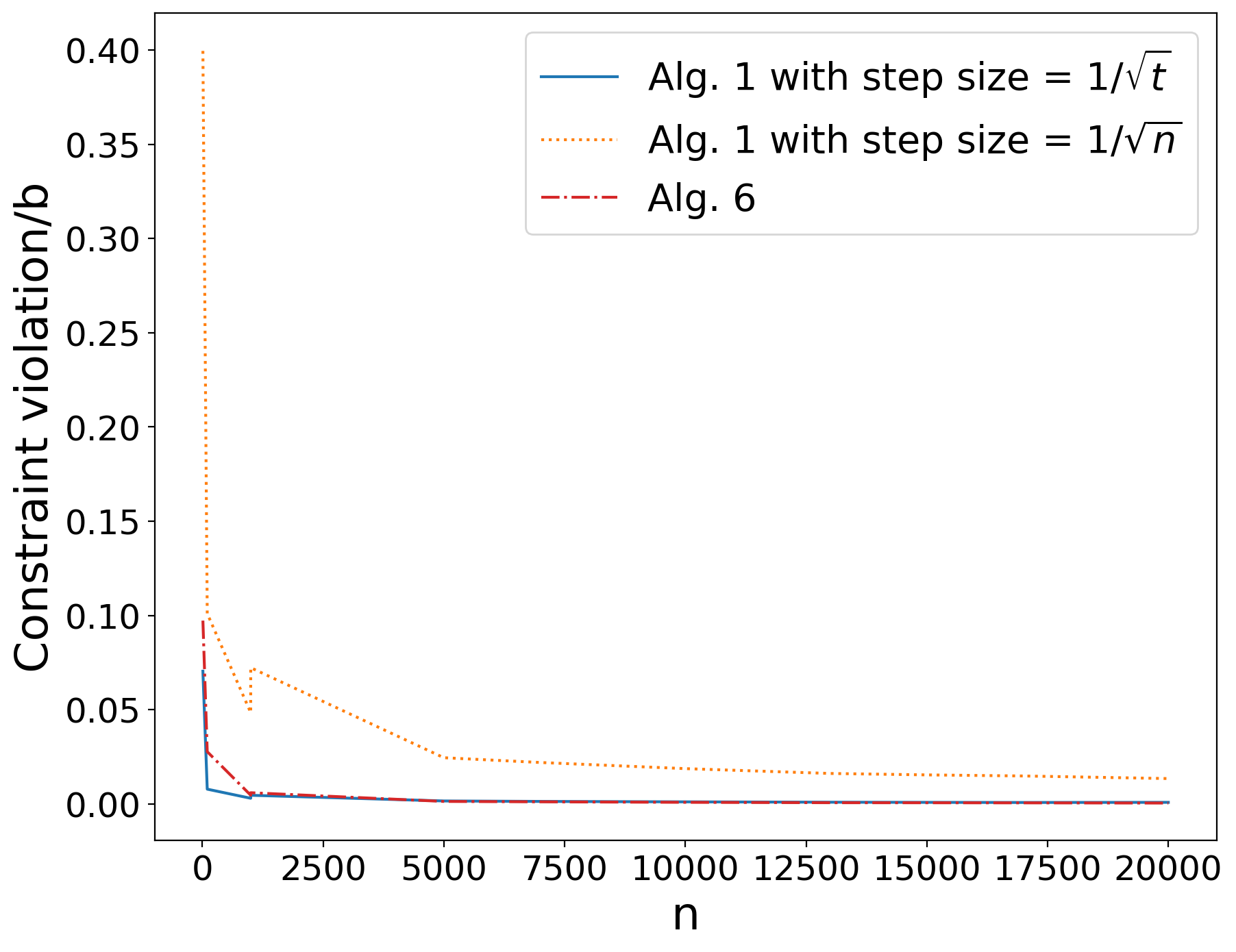}
  \caption{Constraint Violation}
  \label{fig:exp1b}
    \end{subfigure}
    \caption{Experiment with Uniform i.i.d. input}
    \label{exp1}
\end{figure}

In the second experiment (Figure \ref{exp2}), we compare the three algorithms in a setting where the boundedness of the support of distribution $\mathcal{P}$ is violated. Specifically, $m=10,$ $a_{ij}$'s are generated i.i.d. from $\mathcal{N}(1,1)$ and $r_j = \sum_{i=1}^m a_{ij} - \epsilon_j$ where $\epsilon_j\sim \text{Unif}(0,m).$ For each value of $n$, we run $100$ simulation trials, and in each trial, $d_i$'s are sampled i.i.d. from Unif$[1/3,2/3].$ In this experiment, the regret performances of Algorithm \ref{alg:SOA} (with step size of $1/\sqrt{t}$) and Algorithm \ref{alg:SNA} are quite close to each other, while Algorithm \ref{alg:SNA} still performs better in respect with constraint violation. The feasible algorithm (Algorithm \ref{alg:SFA}) still achieves regret on the order of $\sqrt{n}.$ Note that our theoretical results, also all the previous literature on online LP problem, rely on the boundedness assumption for the LP input. An open question is how to generalize these analyses to the case when the distribution $\mathcal{P}$ has an unbounded support.

\begin{figure}[ht!]
\begin{subfigure}{.5\textwidth}
  \centering
  \includegraphics[width=.93\linewidth]{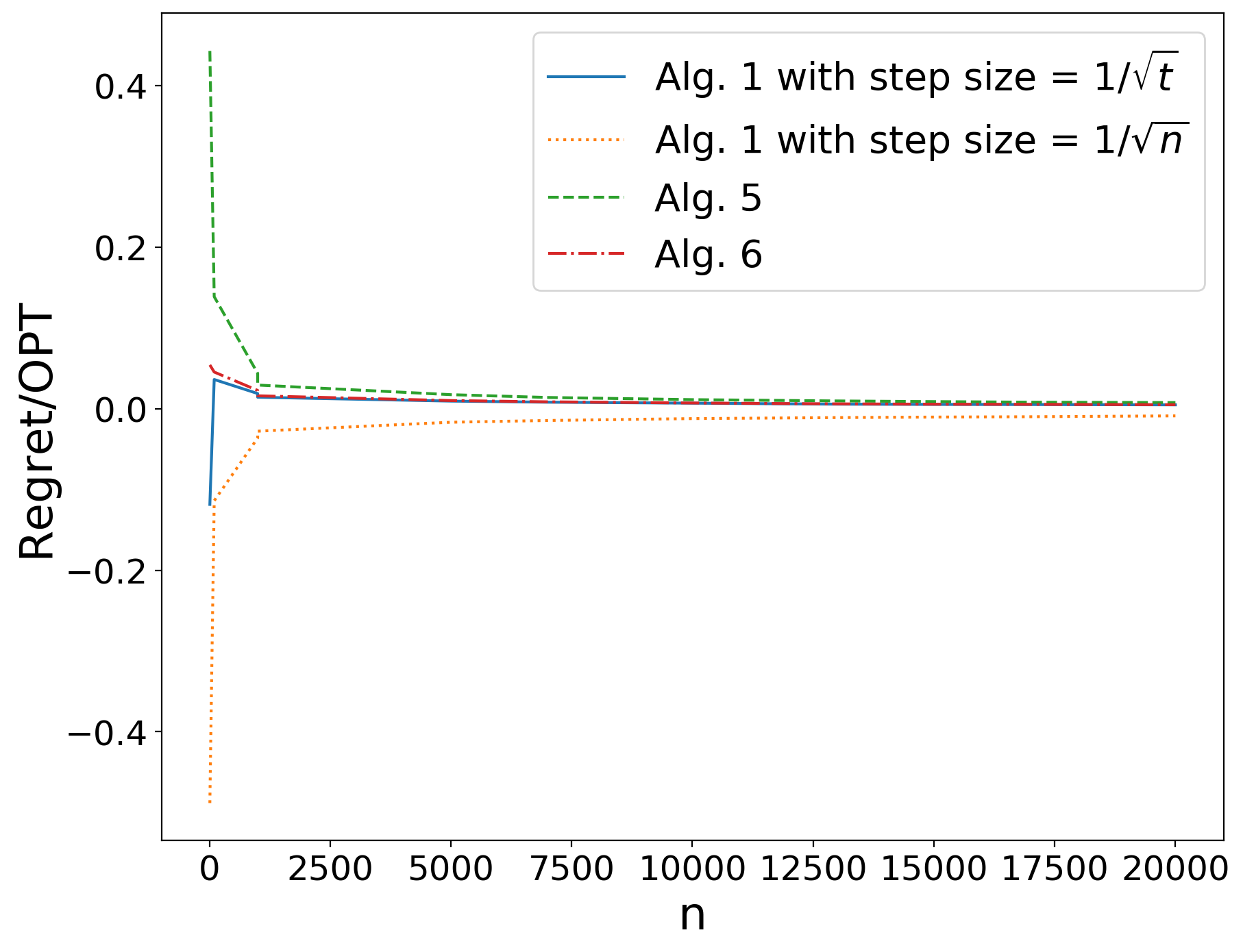}
  \caption{Regret}
  \label{exp2a}
\end{subfigure}%
\begin{subfigure}{.5\textwidth}
  \centering
  \includegraphics[width=.98\linewidth]{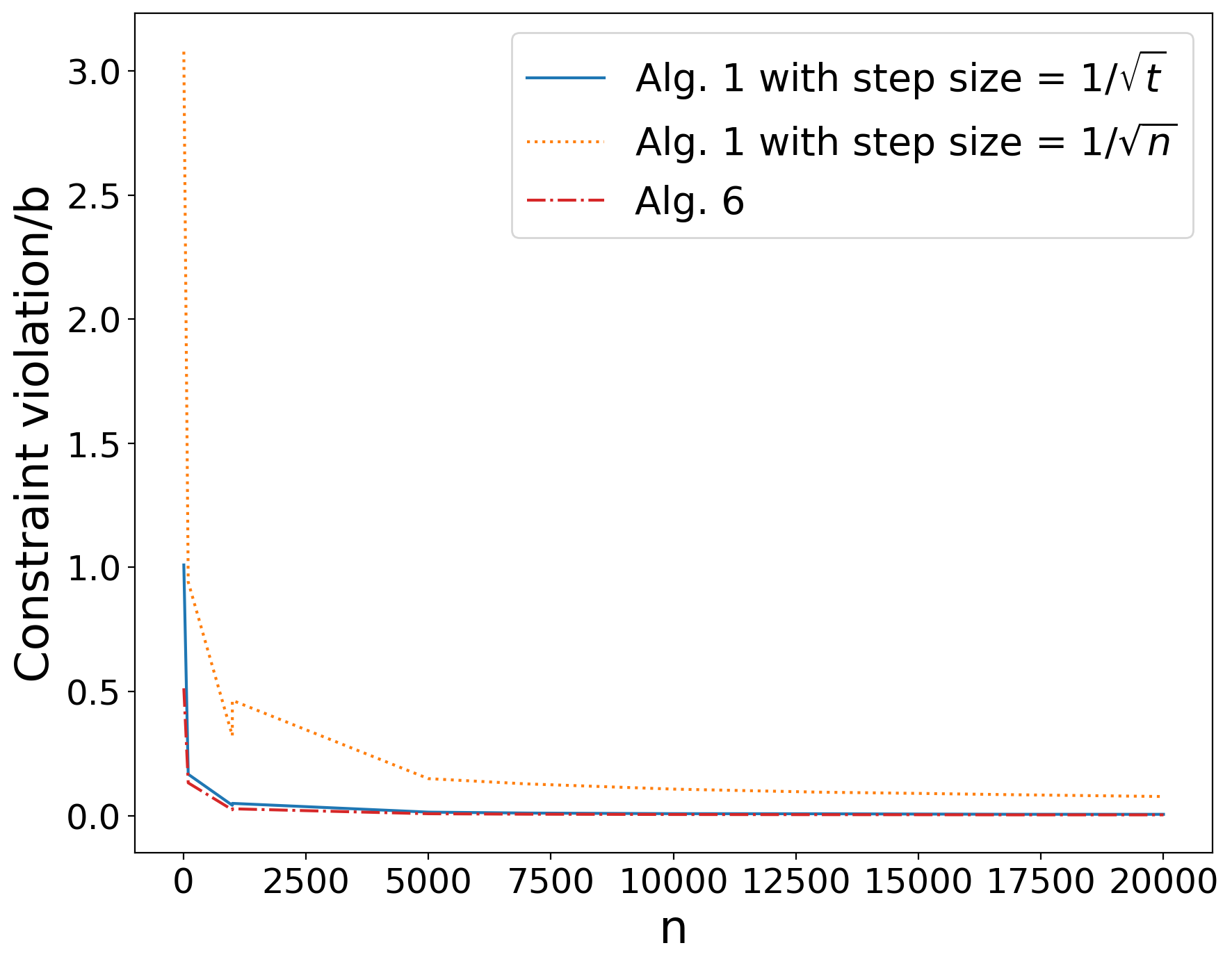}
  \caption{Constraint Violation}
  \label{fig:exp2b}
    \end{subfigure}
    \caption{Experiment with Gaussian i.i.d. input}
    \label{exp2}
\end{figure}

The third experiment (Figure \ref{exp3}) presents a negative result on all three algorithms. Specifically, $a_{ij}$'s are generated i.i.d. from truncated Cauchy$(1,1)$ (with different thresholds) and $r_j = \sum_{i=1}^m a_{ij} - \epsilon_j$ where $\epsilon_j\sim \text{Unif}(0,m).$ As before, for each value of $n$, we run $100$ simulation trials, and in each trial, $d_i$'s are sampled i.i.d. from Unif$[1/3,2/3].$ We observe that the performance becomes unstable as the truncation threshold goes up. The phenomenon is consistent with the previous analysis that the algorithm regret is positively affected by the upper bound on $a_{ij}$'s and $r_j$'s. The empirical finding suggests that a light-tail distribution is probably necessary for an online LP algorithm to succeed.

\begin{figure}[ht!]
\begin{subfigure}{.5\textwidth}
  \centering
  \includegraphics[width=.93\linewidth]{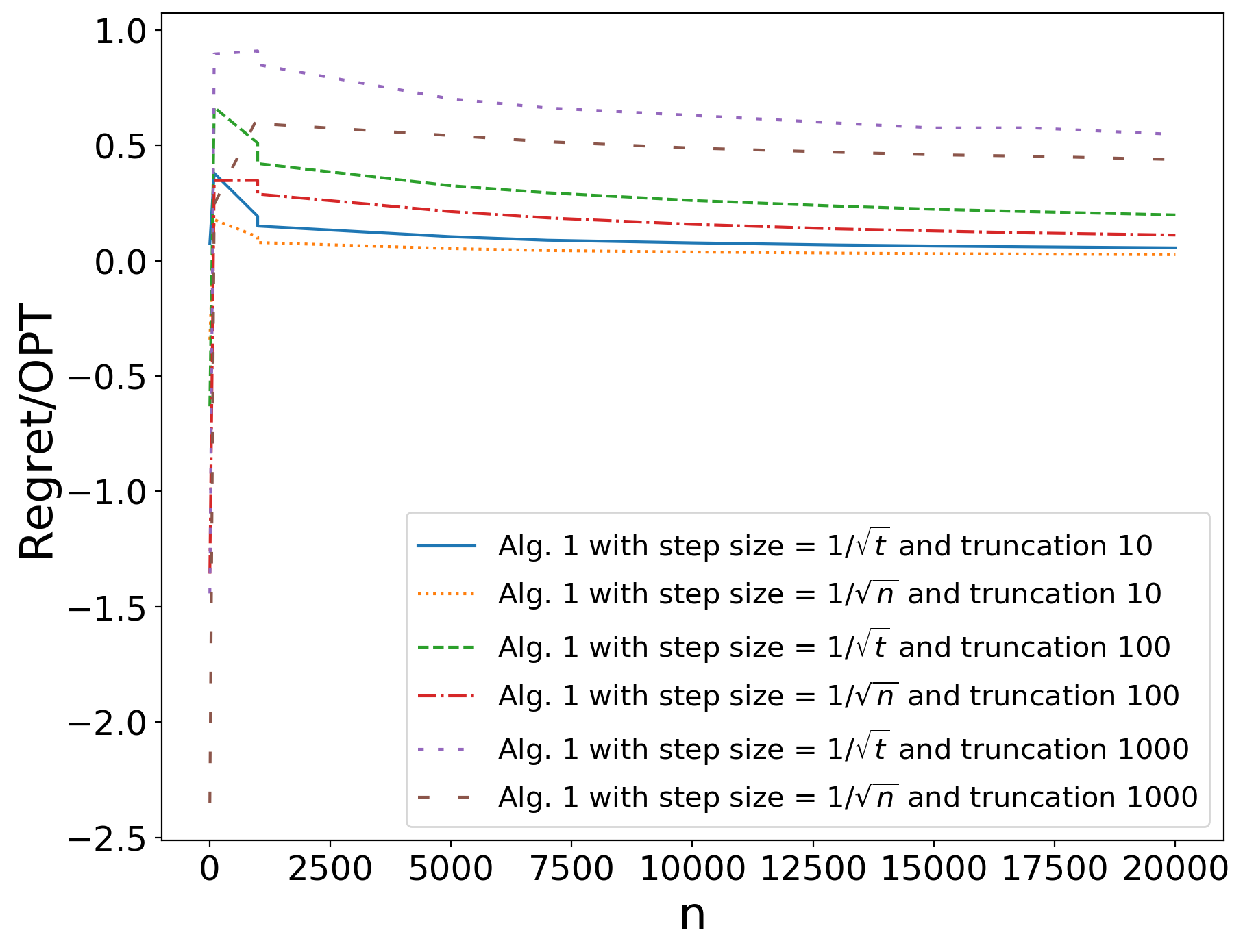}
  \caption{Regret}
  \label{exp3a}
\end{subfigure}%
\begin{subfigure}{.5\textwidth}
  \centering
  \includegraphics[width=.98\linewidth]{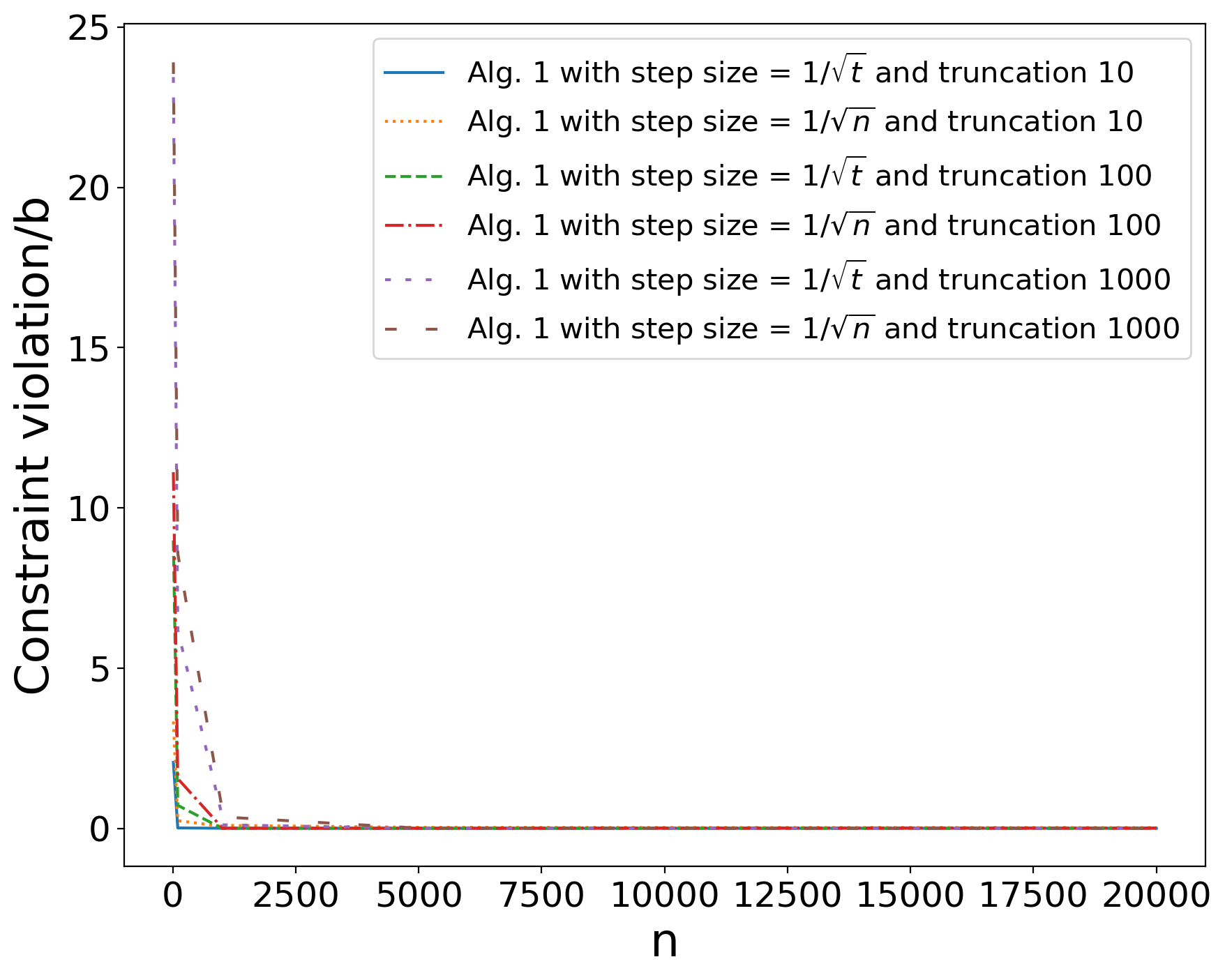}
  \caption{Constraint Violation}
  \label{fig:exp3b}
    \end{subfigure}
    \caption{Experiment with Cauchy i.i.d. input}
    \label{exp3}
\end{figure}

Figure \ref{exp5} presents the algorithm performance under the random permutation model. We first generate four groups of data with equal size from four different distributions and then combine these groups as the LP input: (i) $a_{ij}$'s are generated from Unif$[0,2]$; (ii) $a_{ij}$ are generated from $\mathcal{N}(1,1)$; (iii) $a_{ij}$ are generated from $\mathcal{N}(0,1)$; (iv)  $a_{ij}$ are generated from uniform distribution on $\{-1,1,3\}$. $r_j$'s for all four groups are generated from Unif$[0,1]$. Note this data set can not be generated from any distribution in the stochastic input model. For each value of $n$, we run $100$ simulation trials, and in each trial, $d_i$'s are sampled i.i.d. from Unif$[1/3,2/3].$ In this experiment, we observe that the algorithms all achieves $O(\sqrt{n})$ regret except for Algorithm \ref{alg:SOA} with step size $1/\sqrt{n}$. The step size results in a negative regret but much larger constraint violation than the other algorithms. All the presented algorithms achieve $O(\sqrt{n})$ constraint violation.

\begin{figure}[ht!]
\begin{subfigure}{.5\textwidth}
  \centering
  \includegraphics[width=.93\linewidth]{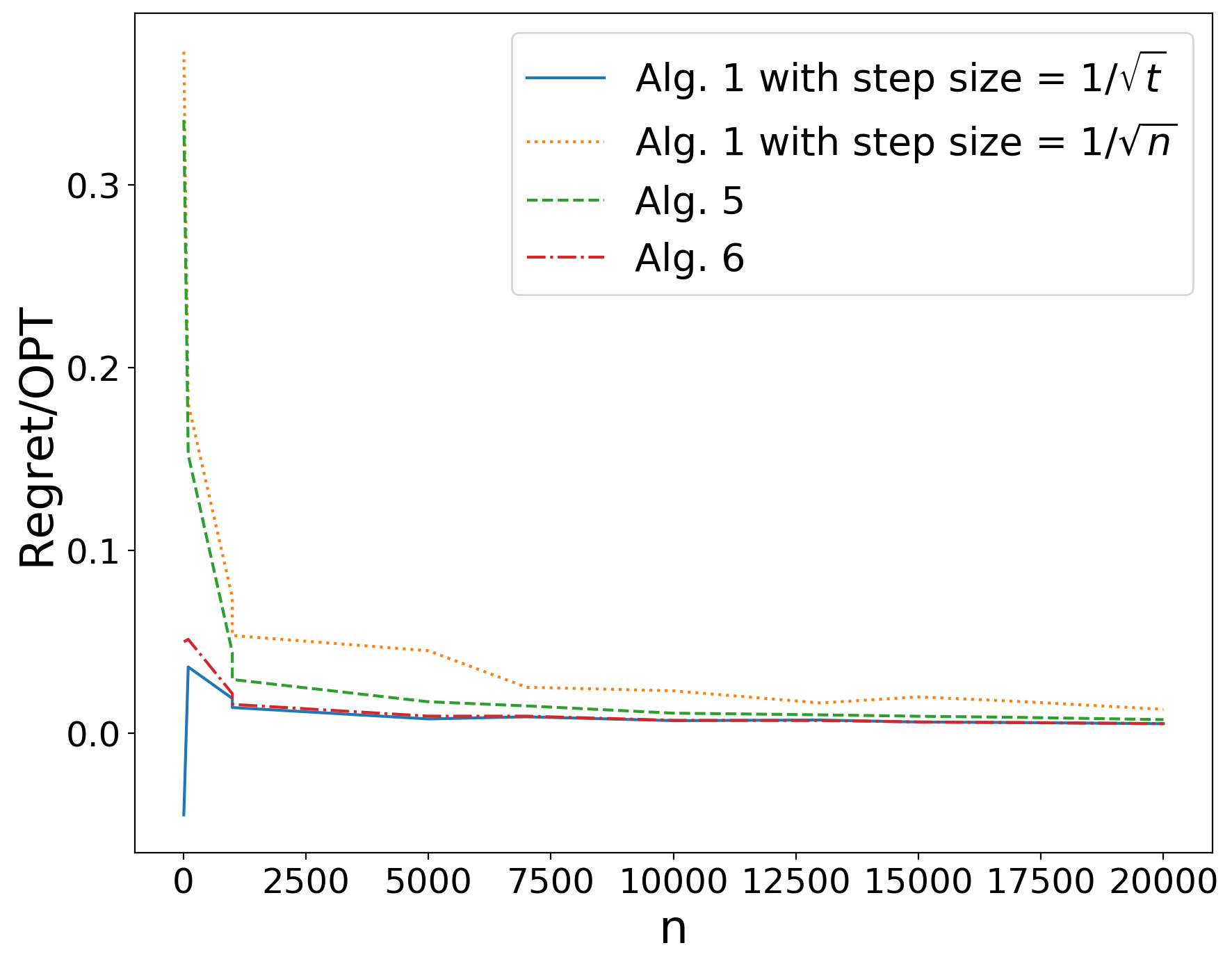}
  \caption{Regret}
  \label{exp5a}
\end{subfigure}%
\begin{subfigure}{.5\textwidth}
  \centering
  \includegraphics[width=.98\linewidth]{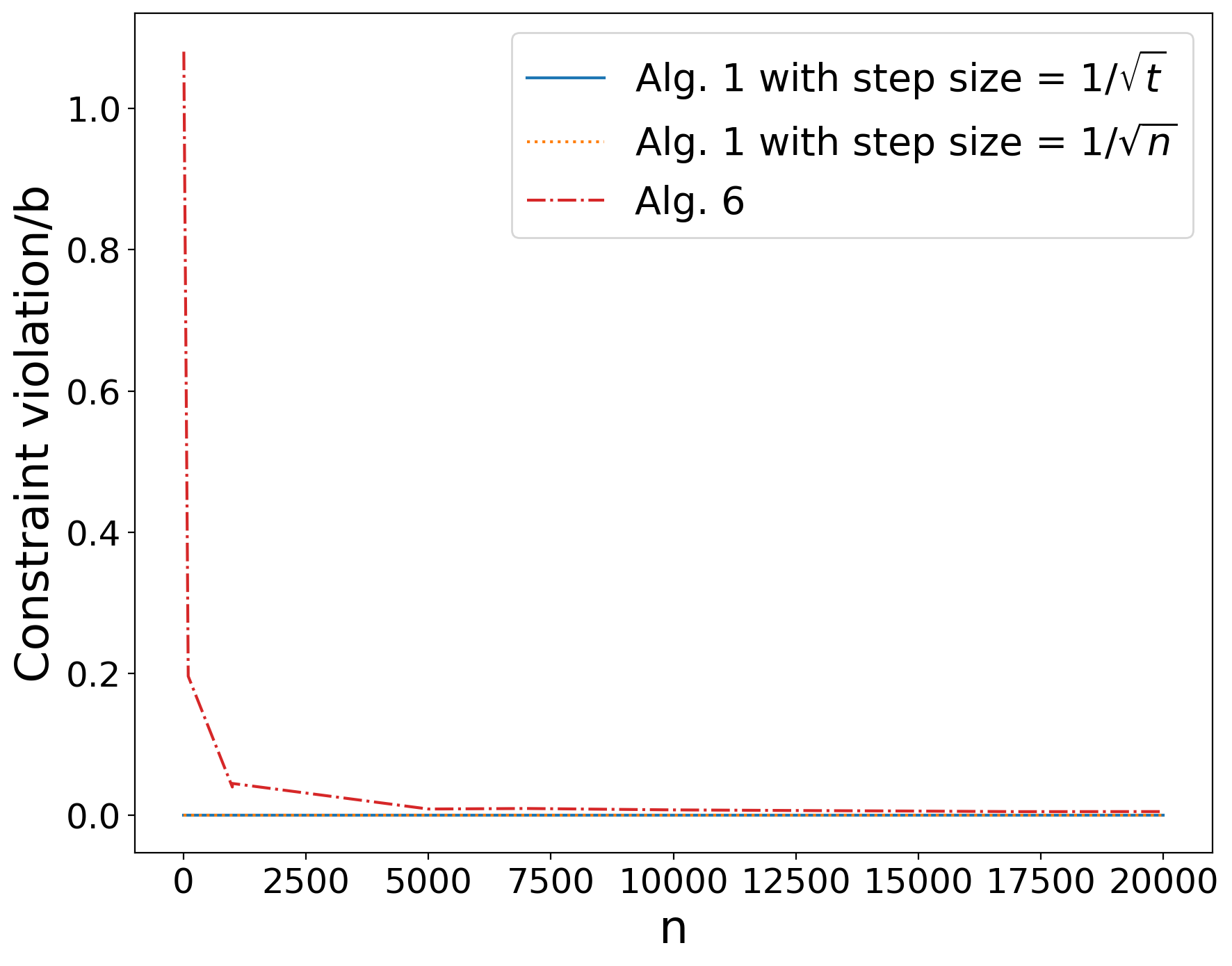}
  \caption{Constraint Violation}
  \label{fig:exp5b}
    \end{subfigure}
    \caption{Experiment with randomly permuted input}
    \label{exp5}
\end{figure}

In addition, we conduct two groups of experiments to illustrate the computational aspect of our algorithms. The algorithms are implemented in Matlab and on a PC with Intel Core i7-9700K Processor. We use the Gurobi solver - one of the state-of-the-art LP and integer LP solvers - as benchmark and to compute the optimal solution. We emphasize that the code for our algorithms, unlike the Gurobi solver, is not fully optimized in the experiment, so there is still great room for improving the computational efficiency of our algorithms. Table \ref{tabRP} presents the experiments under the worst-case example for online LP problem in \citep{agrawal2014dynamic}. The example is constructed under the random permutation model, and it provides a lower bound on the right-hand-side of the constraint for the existence of an $(1-\epsilon)$-competitiveness online LP algorithm. In this sense, it represents one of the most challenging problem instances under the random permutation model. Gurobi solves the binary LP problem with $1\%$ MIPGap and all competitiveness ratios in the table are reported against the optimal objective value of the relaxed LP's optimal objective. Gurobi computes the optimal solution in an offline fashion while the other three algorithms are online.  The numbers are computed based on an average over 100 different problem instances. Algorithm \ref{alg:SOA} is implemented with both step size of $\frac{1}{\sqrt{t}}$ and $\frac{1}{\sqrt{n}}.$ We stop the algorithm and set the rest decision variables to be zero when any of the constraint is exhausted. Though Algorithm \ref{alg:DLA} and Algorithm \ref{alg:PBD} have provably smaller regret bounds, our algorithm provides better empirical performance. For small value of $m$, our algorithm outputs a near-optimal solution within much less time than Gurobi. Also, the competitiveness ratio decreases as $m$ goes larger, which is consistent with the $O(m)$ term in the regret bound of our algorithm.
 
\begin{table}[ht!]
    \centering
    \begin{tabular}{cc|c|c|c|c|c}
    \toprule
         && Gurobi & Alg. \ref{alg:SOA} ($\frac{1}{\sqrt{t}}$) & Alg. \ref{alg:SOA}  ($\frac{1}{\sqrt{n}}$) & Alg. \ref{alg:DLA} & Alg. \ref{alg:PBD}\\ 
         \midrule
   \multirow{2}{*}{\scriptsize $m=8, n=10^3$}  & CPU time  & 0.082 & 0.023 &0.015 & 126.684 & 126.696\\
   & Cmpt. Ratio &100.0\% & 99.1\% & 99.7\% & 89.8\% & 97.9\%\\
\midrule
  \multirow{2}{*}{\scriptsize $m=128, n=10^4$} &  CPU time  & 0.408 & 0.141  & 0.138 & 2338.149 & 2338.285\\ 
   &Cmpt. Ratio & 100.0\% & 99.3\% & 98.8\% & 94.8\% & 97.7\% \\
\midrule
    \multirow{2}{*}{\scriptsize $m=1024, n=10^5$} &  CPU time   &  52.496  & 3.479 & 3.270 & >3000 &  >3000 \\ 
  & Cmpt. Ratio & 100.0\% & 94.6\%& 98.7\%  \\
\midrule
    \multirow{2}{*}{\scriptsize $m=4096, n=10^5$} &  CPU time   &  114.96  & 27.093 & 32.020 & >3000 &  >3000 \\ 
  & Cmpt. Ratio & 99.6\% & 73.1\%& 83.3\%  \\  
\midrule
    \multirow{2}{*}{\scriptsize $m=4096, n=2\times10^5$} &  CPU time   &  254.243  & 57.541 & 53.887 & >3000 &  >3000 \\ 
  & Cmpt. Ratio & 99.7\% & 80.0\%& 89.1\%  \\  
   \bottomrule
    \end{tabular}
    \caption{Performance under worst-case example in \citep{agrawal2014dynamic}}
    \label{tabRP}
\end{table}

In Table \ref{tab1}, we test the performance of our algorithm on some Multi-knapsack benchmark problems \citep{chu1998genetic, drake2016case}. As the last experiment, algorithm \ref{alg:SOA} is implemented with both step size of $\frac{1}{\sqrt{t}}$ and $\frac{1}{\sqrt{n}},$ and the experiment setup is the same as the last experiment. The competitiveness ratios are reported against the relaxed LP's optimal objective value. The computational advantage of our simple online algorithm is significant. Like the last experiment, although Algorithm \ref{alg:DLA} and Algorithm \ref{alg:PBD} reduces the regret upper bound by a factor of $\sqrt{m}$, the advantage of these two algorithms is not evident in practice. It deserves more efforts to understand whether the current regret bound for Algorithm \ref{alg:SOA} is tight. Also, it merits more study on how to design online algorithms that works effectively for a large-$m$-and-small-$n$ regime.

\begin{table}[ht!]
    \centering
    \begin{tabular}{cc|c|c|c|c|c}
    \toprule
         && Gurobi & Alg. \ref{alg:SOA} ($\frac{1}{\sqrt{t}}$) & Alg. \ref{alg:SOA}  ($\frac{1}{\sqrt{n}}$) & Alg. \ref{alg:DLA} & Alg. \ref{alg:PBD} \\ 
         \midrule
   \multirow{2}{*}{\scriptsize $m=5, n=500$}  & time  & 0.116 & 0.006 &0.006 &70&70 \\
   & Cmpt. Ratio &99.6\% & 92.3\% & 75.05\% &91.8\%&91.5\%\\
\midrule
  \multirow{2}{*}{\scriptsize $m=10, n=500$} &  time  & 0.136 & 0.006  & 0.006 & 132 & 132 \\ 
   &Cmpt. Ratio & 99.6\% & 91.8\% & 80.9\%  & 91.6\% & 90.7\%\\
\midrule
    \multirow{2}{*}{\scriptsize $m=30, n=500$} &   time  &  95.2 & 0.006 & 0.005 & 134 & 133\\ 
  & Cmpt. Ratio & 99.4\% & 91.5\%& 89.4\% & 89.1\% & 90.4\% \\
\midrule
    \multirow{2}{*}{\scriptsize $m=10^3, n=10^5$} &   time  & 857 & 2.711 & 2.679 & >3000 & >3000 \\ 
  & Cmpt. Ratio & 99.8\% & 94.9\%& 97.3\% &  &  \\
\midrule
    \multirow{2}{*}{\scriptsize $m=3\times10^3, n=10^5$} &   time  & 1069 & 16.63 & 16.61 & >3000 & >3000 \\ 
  & Cmpt. Ratio & 94.0\% & 82.3\%& 86.8\% &  &  \\
\midrule
    \multirow{2}{*}{\scriptsize $m=3\times10^3, n=2\times10^5$} &   time  & 2799 & 40.21 & 34.84 & >3000 &  >3000 \\ 
  & Cmpt. Ratio & 93.8\% & 84.7\%& 88.4\% &  &  \\
   \bottomrule
    \end{tabular}
    \caption{Multi-knapsack benchmark problem}
    \label{tab1}
\end{table}

\section*{Acknowledgement}

We thank all seminar participants at Stanford, NYU Stern, Columbia DRO, Chicago Booth, and Imperial College Business School for helpful discussions and comments.

\bibliographystyle{informs2014} 
\bibliography{sample.bib} 

\appendix
\section*{Appendix}

\renewcommand{\thesubsection}{A\arabic{subsection}}

\subsection{Concentration Inequalities under Random Permutation}

\begin{lemma}
Let $U_1 , ... , U_n$ be a random sample without replacement from the real numbers $\{c_1,...,c_N\}$. Then for every $s > 0$,
\begin{equation*}
    \begin{split}
        \mathbb{P}(|\bar{U_n}-\bar{c}_N|\geq s)
        \leq
        \left\{
        \begin{matrix*}[l]
            2\exp\left(-\frac{2ns^2}{\Delta_N^2}\right) &\text{(Hoeffding)},\\
            2\exp\left(-\frac{2ns^2}{(1-(n-1)/N)\Delta_N^2}\right) &\text{(Serfling)},\\
            2\exp\left(-\frac{ns^2}{2\sigma^2_N+s\Delta_N}\right) &\text{(Hoeffding-Bernstein)},\\
            2\exp\left(-\frac{ns^2}{m\sigma_N^2}\right) &\text{if $N = mn$ (Massart)},
        \end{matrix*}
        \right.
    \end{split}
\end{equation*}
where $\bar{c}_N = \frac{1}{N}\sum\limits_{i=1}^{N}c_i$, $\sigma^2_N=\frac{1}{N}\sum\limits_{i=1}^{N}(c_i-\bar{c}_N)^2$ and 
$\Delta_N=\max\limits_{1\leq i \leq N}c_i-\min\limits_{1\leq i \leq N}c_i$.
\label{Hoeffding}
\end{lemma}

\begin{proof}
See Theorem 2.14.19 in \cite{van1996weakconvergence}.
\end{proof}

\subsection{Proof of Lemma \ref{iidBound}}
\label{PFiidBound}

\begin{proof}
For $\bm{p}^*$, the optimal solution of (\ref{SP}), we have 
$$\underline{d}\|\bm{p}^*\|_1 \leq  \bm{d}^T\bm{p}^* \overset{(a)}{\leq} \E r\leq \bar{r},$$ 
where inequality (a) is due to that if otherwise, $\bm{p}^*$ cannot be the optimal solution because it will give a larger objective value of $f(\bm{p})$ than setting $\bm{p}=\bm{0}.$ Given the non-negativeness of $\bm{p}^*$,
we have $\|\bm{p}^*\|_2 \le \|\bm{p}^*\|_1.$ So we obtain the first inequality in the lemma.
    
For $\bm{p}_t$ specified by Algorithm \ref{alg:SOA}, we have, 
\begin{align*}
    \|\bm{p}_{t+1}\|^2_2 & \leq
            \left\|\bm{p}_{t}+\gamma_t\left(\bm{a}_tx_t-\bm{d}  \right)\right\|^2_2 \\
            & = \|\bm{p}_{t}\|_2^2 + \gamma_t^2 \|\bm{a}_tx_t - \bm{d}\|_2^2 + 2\gamma_t (\Aa_t x_t - \Dd)^\top \p_{t}\\
            & \le \|\bm{p}_{t}\|_2^2 + \gamma_t^2 m(\bar{a}+\bar{d})^2 + 2\gamma_t \Aa_t^\top \p_{t} x_t - 2\gamma_t\bm{d}^\top \bm{p}_{t}
\end{align*}
where the first inequality comes from the projection (into the non-negative orthant) step in the algorithm. Note that 
$$\bm{a}_t^\top\bm{p}_{t}x_t = \bm{a}_t^\top\bm{p}_{t}I(r_t>\bm{a}_t^\top\bm{p}_{t}) \le r_t \le \bar{r}.$$
Therefore,
\begin{align*}
  \|\bm{p}_{t+1}\|^2_2 & \le \|\bm{p}_{t}\|_2^2 + \gamma_t^2 m(\bar{a}+\bar{d})^2 + 2\gamma_t \bar{r} - 2\gamma_t\bm{d}^\top \bm{p}_{t},
\end{align*}
and it holds with probability $1$. 

Next, we establish that when $\|\bm{p}_{t}\|_2$ is large enough, then it must hold that $\|\bm{p}_{t+1}\|_2 \le \|\bm{p}_{t}\|_2$. Specifically, when $\|\bm{p}_{t}\|_2\geq\frac{2\bar{r}+m(\bar{a}+\bar{d})^2}{\underline{d}}$, we have 
\begin{align*}
   \|\p_{t+1}\|_2^2 - \|\p_t\|_2^2 & \le \gamma_t^2 m(\bar{a}+\bar{d})^2 + 2\gamma_t \bar{r} - 2\gamma_t\bm{d}^\top \bm{p}_{t} \\
   &\leq
  \gamma_t^2 m(\bar{a}+\bar{d})^2 + 2\gamma_t \bar{r} - 2\gamma_t\underline{d}\|\bm{p}_{t}\|_1\\
   &\leq
  \gamma_t^2 m(\bar{a}+\bar{d})^2 + 2\gamma_t \bar{r} - 2\gamma_t\underline{d}\|\bm{p}_{t}\|_2\\
   &\leq0
\end{align*}
when $\gamma_t\le 1.$
On the other hand, when $\|\bm{p}_{t}\|_2\leq\frac{2\bar{r}+m(\bar{a}+\bar{d})^2}{\underline{d}}$, 
\begin{align*}
    \|\bm{p}_{t+1}\|_2 & \leq  \left\|\bm{p}_{t}+\gamma_t\left(\bm{a}_tx_t-\bm{d}  \right)\right\|_2 \\
     & \overset{(b)}{\leq}
    \|\bm{p}_{t}\|_2
    +
    \gamma_t\|\bm{a}_tx_t-\bm{d}\|_2 \\
&     \leq
    {\frac{2\bar{r}+m(\bar{a}+\bar{d})^2}{\underline{d}}}
    +
    m(\bar{a}+\bar{d})
\end{align*}
where (b) comes from the triangle inequality of the norm.

Combining these two scenarios with the fact that $\p_1 = \bm{0}$, we have 
\begin{align*}
    \|\bm{p}_{t}\|_2
    \leq
    {\frac{2\bar{r}+m(\bar{a}+\bar{d})^2}{\underline{d}}}
    +
    m(\bar{a}+\bar{d})
\end{align*}
for $t=1,...,n$ with probability 1.

\end{proof}

\subsection{Proof of Theorem \ref{theoiid}}
\label{PFtheoiid}

\begin{proof}
First, the primal optimal objective value is no greater than the dual objective with $\bm{p}=\bm{p}^*.$ Specifically,
\begin{align*}
    R_n^* = \text{P-LP} & = \text{D-LP}\\
    & \le n\Dd^\top \p^* + \sum_{j=1}^n \left(r_j-\bm{a}_j^\top \bm{p}^*\right)^+. 
\end{align*}
Taking expectation on both sides,
\begin{align*}
   \E\left[R_n^*\right] & \le\E\left[ n\Dd^\top \p^* + \sum_{t=1}^n \left(r_t-\bm{a}_t^\top \bm{p}^*\right)^+\right] \\
   & \le nf(\bm{p}^*).
\end{align*}
Thus, the optimal objective value of the stochastic program (by a factor of $n$) is an upper bound for the expected value of the primal optimal objective. Hence
\begin{align*}
    \E[R_n^* - R_n] & \le nf(\bm{p}^*) - \sum_{j=1}^n \E \left[r_tI(r_t > \Aa_t^\top \p_{t})\right] \\
    & \le \sum_{t=1}^n \E\left[f(\bm{p}_t)\right] - \sum_{t=1}^n \E \left[r_tI(r_t > \Aa_t^\top \p_{t})\right] \\
    & \le \sum_{t=1}^n \E \left[\bm{d}^\top\bm{p}_t + \left(r_t-\bm{a}_t^\top\bm{p}_t\right)^+- r_tI(r_t > \Aa_t^\top \p_{t})\right] \\
    & = \sum_{t=1}^n \E \left[\left(\bm{d}^\top- \Aa_tI(r_t > \Aa_t^\top \p_{t})\right)^\top \bm{p}_t\right].
\end{align*}
where the expectation is taken with respect to $(r_t, \Aa_t)$'s. In above, the second line comes from the optimality of $\bm{p}^*$, while the third line is valid because of the independence between $\bm{p}_t$ and $(r_t, \bm{a}_t).$

The importance of the above inequality lies in that it relates and represents the primal optimality gap in the dual prices $\bm{p}_t$ -- which is the core of Algorithm \ref{alg:SOA}. 
From the updating formula in Algorithm \ref{alg:SOA}, we know
\begin{align*}
    \|\bm{p}_{t+1}\|_2^2 & \le \|\p_{t}\|_2^2 - \frac{2}{\sqrt{n}}\left(\bm{d}- \Aa_tI(r_t > \Aa_t^\top \p_{t})\right)^\top \bm{p}_t + \frac{1}{n} \left\|\bm{d}- \Aa_tI(r_t > \Aa_t^\top \p_{t})\right\|^2_2\\
    & \le \|\p_{t}\|_2^2 - \frac{2}{\sqrt{n}}\left(\bm{d}- \Aa_tI(r_t > \Aa_t^\top \p_{n})\right)^\top \bm{p}_t + \frac{m(\bar{a}+\bar{d})^2}{n}.
\end{align*} 
Moving the cross-term to the right-hand-side, we have
\begin{align*}
    \sum_{t=1}^n \left(\bm{d}- \Aa_tI(r_t > \Aa_t^\top \p_{t})\right)^\top \bm{p}_t & \le \sum_{t=1}^n \left(\sqrt{n}\|\p_t\|_2^2 - \sqrt{n}\|\p_{t+1}\|_2^2 + \frac{m(\bar{a}+\bar{d})^2}{\sqrt{n}}\right)\\
    & \le 
    m(\bar{a}+\bar{d})^2\sqrt{n}.
\end{align*}
Consequently,
$$\E[R_n^* - R_n] \le   m(\bar{a}+\bar{d})^2\sqrt{n}$$
hold for all $n$ and distribution $\mathcal{P}\in \Xi.$

For the constraint violation, if we revisit the updating formula, we have
$$\bm{p}_{t+1} \ge \bm{p}_t + \frac{1}{\sqrt{n}}\left(\bm{a}_tx_t-\bm{d}\right)$$
where the inequality is elementwise. Therefore, 
\begin{align*}
\sum_{t=1}^n \bm{a}_tx_t & \le n\bm{d} + \sum_{t=1}^n \sqrt{n}(\bm{p}_{t+1}-\bm{p}_{t}) \\
& \le \bm{b} + \sqrt{n}\bm{p}_{n+1} 
\end{align*}
In the second line, we remove the term involve $\p_1$ with the algorithm specifying $\p_1=\bm{0}.$ Then with Lemma \ref{iidBound}, we have 
$$ \E \left[v(\bm{x})\right]= \E \left[ \|\left(\bm{A}\bm{x}-\bm{b}\right)^+\|_2\right] \le
\sqrt{n}\E \|\bm{p}_{n+1}\|_2  \leq \left({\frac{2\bar{r}+m(\bar{a}+\bar{d})^2}{\underline{d}}} + m(\bar{a}+\bar{d})\right) \sqrt{n}.$$
\end{proof}

\subsection{Proof of Proposition \ref{importantLemma}}
\label{PFimportant}
\begin{proof}
Define SLP$(s, \bm{b}_0)$ as the following LP
\begin{align*}
   \max \ \ & \sum_{j=1}^s r_jx_j   \\
    \text{s.t. }\ & \sum_{j=1}^s a_{ij}x_j \le \frac{sb_i}{n} + b_{0i} \\
    & 0 \le x_j \le 1\ \text{ for } j=1,...,s.\nonumber
\end{align*}
where $\bm{b}_0 = (b_{01},...,b_{0m})$ denotes the constraint relaxation quantity for the scaled LP. Denote the optimal objective value of SLP$(s, \bm{b}_0)$ as $R^*(s, \bm{b_0}).$ Also, denote $\bm{x}(\bm{p}) = (x_1(\bm{p}),...,x_n(\bm{p}))$ and $x_j(\bm{p})=I(r_j>\bm{a}_j^\top\bm{p}).$ It denotes the decision variables we obtain with a dual price $\bm{p}.$

We prove the following three results:
\begin{itemize}
    \item[(i)] The following bounds hold for $R_n^*,$ 
    $$\sum_{j=1}^n r_jx_j(\bm{p}_n^*) \le R_n^* \le \sum_{j=1}^n r_jx_j(\bm{p}_n^*) + m\bar{r}.$$
    \item[(ii)] When $s\ge \max\{16\bar{a}^2,\exp{\{16\bar{a}^2\}},e\},$ then the optimal dual solution $\bm{p}_n^*$ is a feasible solution to SLP$\left(s, \frac{\log s}{\sqrt{s}}\bm{1}\right)$ with probability no less than $1-\frac{m}{s}$.
    \item[(iii)] The following inequality holds for the optimal objective values to the scaled LP and its relaxation
    $$R_s^* \ge R^*\left(s, \frac{\log s}{\sqrt{s}} \bm{1} \right) - \frac{\bar{r}\sqrt{s}\log s }{\underline{d}}.$$
\end{itemize}

\textbf{For part (i)}, this inequality replace the optimal value with bounds containing the objective values obtained by adopting optimal dual solution. The left hand side of the inequality comes from the complementarity condition while the right hand side can be shown from Lemma \ref{gp}.

\textbf{For part (ii)}, the motivation to introduce a relaxed form of the scaled LP is to ensure the feasibility of $\bm{p}_n^*$. The key idea for the proof is to utilize the feasibility of $\bm{p}_n^*$ for (\ref{eqn:P-LP}). To see that, let $\alpha_{ij}=a_{ij}I(r_j>\bm{a}_j^T\bm{p}^*)$ and 
    \begin{equation}
        \begin{split}
            c_\alpha
            & =
         \max\limits_{i,j} \alpha_{ij}
            -
         \min\limits_{i,j}\alpha_{ij}\leq 2\bar{a},\\
            \bar{\alpha}_i
            &=
         \frac{1}{n}\sum\limits_{j=1}^{n} \alpha_{ij}
         =\frac{1}{n}\sum\limits_{j=1}^{n} a_{ij}x_t(\bm{p}_n^*)
            \leq 
            d_i, \label{ineqP}\\
            \sigma^2_i
            &=
          \frac{1}{n}\sum\limits_{j=1}^{n}(\alpha_{ij}-\bar{\alpha}_i)^2
            \leq 
            4\bar{a}^2. 
        \end{split}
    \end{equation}
 Here the first and third inequality comes from the bounds on $a_{ij}$'s while the second one comes from the feasibility of the optimal solution for (\ref{eqn:P-LP}). 

    Then, when $k>\max\{16\bar{a}^2,\exp{\{16\bar{a}^2\}},e\}$, by applying Hoeffding-Bernstein’s Inequality 
    \begin{equation*}
        \begin{split}
            \mathbb{P}
            \left(
             \sum\limits_{j=1}^k \alpha_{ij}
                -
                kd_i
                \geq
                \sqrt{k}\log k
            \right)
            & \overset{(e)}{\leq}
            \mathbb{P}
            \left(
             \sum\limits_{j=1}^k \alpha_{ij}
                -
                k\bar{\alpha}_i
                \geq
                \sqrt{k}\log k
            \right)\\
            &\overset{(f)}{\leq}
            \exp\left(
         -\frac{k\log^2 k}{8k\bar{a}^2+2\bar{a}\sqrt{k}\log k}
            \right)\\
            &\overset{(g)}{\leq}
            \frac{1}{k}
        \end{split}
    \end{equation*}
    for $i=1,...,m$.  Here inequality (e) comes from (\ref{ineqP}), (f) comes from applying Lemma \ref{Hoeffding}, and (g) holds when $s>\max\{16\bar{a}^2,\exp{\{16\bar{a}^2\}},e\}$.
    
    Let event $$E_i=\left\{\sum_{j=1}^s \alpha_{ij}-sd_i<\sqrt{s}\log s\right\}$$
    and $E=\bigcap\limits_{i=1}^{m}E_i$. The above derivation tells $\prob(E_i) \ge 1-\frac{1}{s}$ By applying union bound, we obtain $\mathbb{P}(E)\geq 1-\frac{m}{s}$ and it completes the proof of part (ii).

  \textbf{For part(iii)}, denote the optimal solution to SLP$\left(s, \frac{\log s}{\sqrt{s}}\bm{1}\right)$ as $\tilde{\bm{p}}_s.$    
    \begin{equation*}
        \begin{split}
            R^*\left(s, \frac{\log s}{\sqrt{s}} \bm{1} \right)
            &=
         s\left(\bm{d}+\frac{\log s}{\sqrt{s}}\bm{1}\right)^\top\tilde{\bm{p}}^*_{s}
            +
            \sum_{j=1}^s\left(r_j-\bm{a}_j^\top\tilde{\bm{p}}^*_{s}\right)^{+}\\
            &\leq
             s\left(\bm{d}+\frac{\log s}{\sqrt{s}}\bm{1}\right)^\top{\bm{p}}^*_{s}
            +
          \sum_{j=1}^s\left(r_j-\bm{a}_j^\top{\bm{p}}^*_{s}\right)^{+} \\
            &\leq
          \frac{\bar{r}\sqrt{s}\log s}{\underline{d}}
            +
          R_s^*.
        \end{split}
    \end{equation*}
    where the first inequality comes from dual optimality of $\tilde{\bm{p}}_s^*$ and the second inequality comes from the upper bound of $\|\bm{p}_s^*\|$ and the duality of the scaled LP $R_s^*$. Therefore,
          $$R_s^* \ge R^*\left(s, \frac{\log s}{\sqrt{s}} \bm{1} \right) - \frac{\bar{r}\sqrt{s}\log s }{\underline{d}}.$$
          
Finally, we complete the proof with the help of the above three results.
\begin{align*}
    \frac{1}{s}\E\left[\mathbb{I}_ER_s^*\right] &\ge \frac{1}{s}\E\left[\mathbb{I}_ER^*\left(s, \frac{\log s}{\sqrt{s}} \bm{1} \right)\right] - \frac{\bar{r}\sqrt{s}\log s }{\underline{d}} \\
    & \ge \frac{1}{s}\E\left[\mathbb{I}_E\sum_{j=1}^sr_jx_j(\bm{p}^*)\right] - \frac{\bar{r}\sqrt{s}\log s }{\underline{d}} \end{align*}
    where $\mathbb{I}_E$ denotes an indicator function for event $E$. The first line comes from applying part (iii) while the second line comes from the feasibility of $\bm{p}^*$ on event $E.$ Then,
    \begin{align*}
    \frac{1}{s}\E\left[R_s^*\right]  &\ge \frac{1}{s}\E\left[\sum_{j=1}^sr_jx_j(\bm{p}^*)\right] - \frac{\bar{r}\sqrt{s}\log s }{\underline{d}} - \frac{m\bar{r}}{s} \\
    & = \frac{1}{n}\E\left[\sum_{j=1}^n r_jx_j(\bm{p}^*)\right] - \frac{\bar{r}\sqrt{s}\log s }{\underline{d}} - \frac{m\bar{r}}{s} \\
    & \ge \frac{1}{n} R_n^* - \frac{\bar{r}\sqrt{s}\log s }{\underline{d}} - \frac{m\bar{r}}{s} - \frac{m\bar{r}}{n}
\end{align*}
where the first line comes from part (ii) -- the probability bound on event $E$, the second line comes from the symmetry of the random permutation probability space, and the third line comes from part (i). We complete the proof.
\end{proof}

\subsection{Proof of Theorem \ref{theoPermut}}

\label{PFtheoPermut}

\begin{proof}
For the regret bound,
\begin{align*}
    R_n^* - \E\left[R_n\right] & =  R_n^* - \sum_{t=1}^n \E\left[r_t x_t\right]
\end{align*}
where $x_t$'s are specified according to Algorithm \ref{alg:SOA}.
Then
\begin{align}
    R_n^* - \E\left[R_n\right] & =   R_n^* - \sum_{t=1}^n \frac{1}{t} \E\left[R_t^*\right]+
    \sum_{t=1}^n \frac{1}{t} \E\left[R_t^*\right]- \sum_{t=1}^n \E\left[r_t x_t\right] \nonumber \\
    & = \sum_{t=1}^n\left(\frac{1}{n}R_n^*-\frac{1}{t} \E\left[R_t^*\right]\right) + \sum_{t=1}^n\E\left[\frac{1}{n+1-t}\tilde{R}_{n-t+1}^* - r_tx_t\right]\label{twoParts}
\end{align}
where $\tilde{R}_{n-t+1}^*$ is defined as the optimal value of the following LP
\begin{align*}
   \max \ \ & \sum_{j=t}^n r_j x_j   \\
    \text{s.t. }\ & \sum_{j=t}^n a_{ij}x_j \le \frac{(n-t+1)b_i}{n}  \\
    & 0 \le x_j \le 1\ \text{ for } j=1,...,m.\nonumber
\end{align*}
For the first part of (\ref{twoParts}), we can apply Proposition \ref{importantLemma}. Meanwhile, the analyses of the second part takes a similar form as the previous stochastic input model. Specifically,  
$$\E\left[\frac{1}{n+1-t}\tilde{R}_{n-t+1}^* - r_tx_t\right] \le \left(\bm{d}- \Aa_tI(r_t > \Aa_t^\top \p_{t})\right)^\top \bm{p}_t.$$
Similar to the stochastic input model,
   \begin{align*}
        \|\bm{p}_{t+1}\|_2^2 & \le \|\p_{t}\|_2^2 - \frac{2}{\sqrt{n}}\left(\bm{d}- \Aa_tI(r_t > \Aa_t^\top \p_{t})\right)^\top \bm{p}_t + \frac{1}{n} \left\|\bm{d}- \Aa_tI(r_t > \Aa_t^\top \p_{t})\right\|^2_2\\
        & \le \|\p_{t}\|_2^2 - \frac{2}{\sqrt{n}}\left(\bm{d}- \Aa_tI(r_t > \Aa_t^\top \p_{t})\right)^\top \bm{p}_t + \frac{m(\bar{a}+\bar{d})^2}{n}.
    \end{align*} 
    Thus, we have 
    \begin{equation*}
        \begin{split}
            \sum\limits_{t=1}^{n}\E \left[\left(\bm{d}- \Aa_tI(r_t > \Aa_t^\top \p_{t})\right)^\top \bm{p}_t\right]
            &\leq
            \sum\limits_{t=1}^{n}\E\left[\sqrt{n}(\|\bm{p}_t\|^2_2 -\|\bm{p}_{t+1}\|^2_2)\right]+\sum\limits_{t=1}^{n}\frac{m(\bar{a}+\bar{d})^2}{\sqrt{n}}\\
            &\leq
            m(\bar{a}+\bar{d})^2\sqrt{n}.
        \end{split}
    \end{equation*}
    Combine two parts above, finally we have
    \begin{equation*}
        \begin{split}
            R_n^* - \mathbb{E}[R_n(\pi)]  
            &\leq
            m\bar{r}
            +
            \frac{\bar{r}\log n\sqrt{n}}{\underline{d}}
            +
            m\bar{r}\log n
            +
            \frac{\max\{16\bar{a}^2,\exp{\{16\bar{a}^2\}},e\}\bar{r}}{n}+m(\bar{a}+\bar{d})^2\sqrt{n}\\
            & = O((m+\log n)\sqrt{n})\end{split}
    \end{equation*}  
    Thus, we complete the proof for the regret. The proof for the constraint violation part follows exactly the same way as the stochastic input model. 
\end{proof}

\subsection{Proof for Theorem \ref{theoMD}}

\begin{proof}
The proof follows mostly the proof of Theorem \ref{theoiid} and Theorem \ref{theoPermut}. We only highlight the difference here. First, the sample average approximation form of the dual problem takes a slightly different form but it is still convex in $\bm{p}.$
\begin{align}
\tag{multi-D-SAA} \min_{\bm{p}}\ & f_n(\bm{p})  = \bm{d}^\top \bm{p} + \frac{1}{n} \sum_{j=1}^n \left(\max_{s=1,...,k}\left\{r_{js}-\bm{a}_{js}^\top \bm{p}\right\}\right)^+  \label{mSAA} \\
\text{s.t. }\ & \bm{p}\ge \bm{0}. \nonumber
\end{align}

The updating formula for $\bm{p}_t$ is different but we can achieve the same relation between $\bm{p}_t$ and $\bm{p}_{t+1}.$ At time $t$, if $\max\limits_{l=1,...,k}\left\{r_{jl}-\bm{a}_{jl}^\top \bm{p}\right\}>0$, we have
\begin{align*}
    \|\bm{p}_{t+1}\|^2_2 & \leq
            \left\|\bm{p}_{t}+\frac{1}{\sqrt{n}}\left(\bm{A}_tx_t-\bm{d}  \right)\right\|^2_2 \\
            & = \left\|\bm{p}_{t}+\frac{1}{\sqrt{n}}\left(\bm{a}_{tl_t}-\bm{d}  \right)\right\|^2_2 \\
            & = \|\bm{p}_{t}\|_2^2 + \frac{1}{n} \|\bm{a}_{tl_t}x_t - \bm{d}\|_2^2 + \frac{2}{\sqrt{n}} (\bm{a}_{tl_t} x_t - \Dd)^\top \p_{t}\\
            & \le \|\bm{p}_{t}\|_2^2 + \frac{m(\bar{a}+\bar{d})^2}{n} + \frac{2}{\sqrt{n}} \Aa_t^\top \p_{t} x_t - \frac{2}{\sqrt{n}}\bm{d}^\top \bm{p}_{t}\\
            &\leq \|\bm{p}_{t}\|_2^2 + \frac{m(\bar{a}+\bar{d})^2}{n} + \frac{2\bar{r}}{\sqrt{n}}- \frac{2}{\sqrt{n}}\bm{d}^\top \bm{p}_{t},
\end{align*}
while if $\max\limits_{l=1,...,k}\left\{r_{jl}-\bm{a}_{jl}^\top \bm{p}\right\}\leq0$, we have
\begin{align*}
    \|\bm{p}_{t+1}\|^2_2 & \leq
            \left\|\bm{p}_{t}+\frac{1}{\sqrt{n}}\left(\bm{A}_tx_t-\bm{d}  \right)\right\|^2_2 \\
            & = \left\|\bm{p}_{t}-\frac{1}{\sqrt{n}}\bm{d}\right\|^2_2 \\
            & \le \|\bm{p}_{t}\|_2^2 + \frac{m(\bar{a}+\bar{d})^2}{n} - \frac{2}{\sqrt{n}}\bm{d}^\top \bm{p}_{t}.
\end{align*}
Combining those two parts, we obtain
\begin{align*}
    \|\bm{p}_{t+1}\|^2_2 & 
            \leq \|\bm{p}_{t}\|_2^2 + \frac{m(\bar{a}+\bar{d})^2}{n} + \frac{2\bar{r}}{\sqrt{n}}- \frac{2}{\sqrt{n}}\bm{d}^\top \bm{p}_{n},
\end{align*}
which is the same formula as the one-dimensional setting. With the above results, the rest of the proof simply follows the same approach as the one-dimensional case.  

\end{proof}

\subsection{Proofs for Results in Section \ref{permutRC}}

\label{ProofPRC}

\subsubsection{Proof for Lemma \ref{lemma:bound_diff_between_test_and_train}}

\begin{proof}
We have
\begin{equation*}
\begin{split}
    \mathbb{E}\left[\sup\limits_{f\in\mathcal{F}}\left|
    \bar{f}(\Z_t) - \bar{f}(\tilde{\Z}_{t'})
    \right| \Big| \Z_n\right]
    &=
    \mathbb{E}\left[\sup\limits_{f\in\mathcal{F}}\left|
    \mathbb{E}\left[\bar{f}(\Z_{t-s})|\Z_t\right]- \mathbb{E}\left[\bar{f}(\tilde{\Z}_s)|\tilde{\Z}_{t'}\right]
    \right|  \Big| \Z_n \right]  \\
    &\leq
    \mathbb{E}\left[\sup\limits_{f\in\mathcal{F}}\left|
    \bar{f}(\Z_{t-s}) -  \bar{f}(\tilde{\Z}_s)
    \right|\Big| \Z_n \right]
    =
    \mathbb{E} \left[
    {Q}_{t,s}(\mathcal{F},\Z_t)\big|\Z_n\right].
\end{split}
\end{equation*}
For the first line, on the right hand side, the two inner expectations are taken with respect to a uniform random sampling on $\Z_t$ and $\tilde{\Z}_{t'}$ respectively. Specifically, $\Z_{t-s}$ (or $\tilde{\Z}_{s}$) can be viewed as a random sampled subset from $\Z_t$ (or $\tilde{\Z}_{t'}$).  For the second line, the first part comes from Jensen's inequality and the expectation in the second part is taken with respect to the random sampling of $\Z_t$ from $\Z_n.$
\end{proof}

\subsubsection{Proof for Lemma \ref{lemma:PRC}}

\begin{proof}
For the first inequality, we refer to Theorem 3 in \citep{tolstikhin2015permutational}. For the second inequality, it is a direct application of Massart's Lemma (See Lemma 26.8 of \citep{shalev2014understanding}).
\end{proof}

\subsubsection{Proof for Proposition \ref{prop_permut_rdm}}

\begin{proof}
Let
$\mathcal{F}_{\bm{p}} 
=
\left\{ 
f_{\bm{p}}: f_{\bm{p}}(r,\bm{a}) = rI\left( r>\bm{a}^T\bm{p}  \right)\right\}$, $\Z_t = \{(r_1,\Aa_1),...,(r_t,\Aa_t)\},$ and $\tilde{Z}_{n-t} = \{(r_{t+1},\Aa_{t+1}),...,(r_n,\Aa_n)\}.$ Also, we assume $n-t>t$ without loss of generality. 
Then,
\begin{align*}
    \E \left[ \left\vert \frac{1}{n-t}\sum_{j=t+1}^n r_{j}I(r_j>\Aa_j^\top \p)  - 
\frac{1}{t}\sum_{j=1}^t r_{j}I(r_j>\Aa_j^\top \p)  \right\vert\right] & \le \mathbb{E} \left[\sup\limits_{f\in\mathcal{F}_{\p}}\left|
\bar{f}(\Z_t) - \bar{f}(\tilde{\Z}_{n-t})
\right|\Big|\Z_n = \mathcal{D}\right]\\
& \le \mathbb{E}\left[
Q_{t,\floor{t/2}}(\mathcal{F},\Z_{t})\Big| \Z_n \right] \\
& \le \frac{4\bar{r}}{t} + \frac{2\sqrt{2\bar{r}^2m\log n}}{\sqrt{t}}.
\end{align*}
Here the first line comes from taking maximum over $\mathcal{F}_{\p}$, the second line comes from lemma \ref{lemma:bound_diff_between_test_and_train} and the third line comes from lemma \ref{lemma:PRC}. 

Similarly, we can show that the inequality on $a_{ij}$'s holds. Thus the proof is completed.
\end{proof}

\subsubsection{Proof for Theorem \ref{theorem_DLA}}

\begin{proof}
At time $t+1,$
\begin{equation*}
    \begin{split}
        \mathbb{E}\left[r_{t+1}x_{t+1}\right]
        &= \mathbb{E}\left[r_{t+1}I(r_{t+1}>\Aa_{t+1}^\top \p_{t+1})\right]
        \\ &=\frac{1}{n-t}
        \mathbb{E}\left[\sum_{j=t+1}^n r_jI(r_{j}>\Aa_j\p_{t+1})\right]\\
        &\geq
        \frac{1}{t}\mathbb{E}\left[
        \sum_{j=1}^t r_jI(r_{j}>\Aa_j\p_{t+1})
        \right] 
        - \frac{4\bar{r}}{\sqrt{\min\{t,n-t\}}}
        -\frac{2\sqrt{2\bar{r}^2m\log n}}{\sqrt{\min\{t,n-t\}}},
    \end{split}
\end{equation*}
where the expectation is taken with respect to the random permutation. The first line comes from the algorithm design, the second line comes from the symmetry over the last $n-t$ terms, and the last line comes from the application of Proposition \ref{prop_permut_rdm}. To relate the first term in the last line with the offline optimal $R_n^*,$ we utilize Proposition \ref{importantLemma}. Then the optimality gap of Algorithm \ref{alg:DLA} is as follows,
\begin{equation*}
    \begin{split}
        R_n^* - \E\left[\sum_{t=1}^n r_t x_t\right] 
        & =  
        R_n^* - \sum_{t=1}^n \E\left[r_t x_t\right]\\
        & \le 
        R_n^* - \sum\limits_{t=2}^{n} \left(\frac{1}{t}\E\left[
        \sum_{j=1}^t r_jI(r_{j}>\Aa_j\p_t)
        \right]
        -
        \frac{4\bar{r}+2\sqrt{2\bar{r}^2m\log n}}{\sqrt{\min\{t,n-t\}}}
        \right)
        \\
        & \leq
        m\bar{r} 
        + 
        \frac{\bar{r}}{\underline{d}}\sqrt{n}\log n
        +
        m\bar{r}\log n
        +
        \bar{r}\max\{16\bar{a}^2,e^{16\bar{a}^2},e\}
        +
        \left(
        8\bar{r}+
        4\sqrt{2\bar{r}^2m\log n}
        \right)\sqrt{n} = O(\sqrt{mn}\log n)
    \end{split}
\end{equation*}
where the last line comes from an application of Proposition \ref{importantLemma}.
Next, we analyze the constraint; again, from Proposition \ref{prop_permut_rdm}, we know 
\begin{align*}
    \frac{1}{n-t}\E \left[\sum_{j=t+1}^n a_{ij}I(r_j>\Aa_j^\top \p_t)\right] & \le 
\E\left[\frac{1}{t}\sum_{j=1}^t a_{ij}I(r_j>\Aa_j^\top \p_t) \right] + \frac{4\bar{a}}{\sqrt{\min\{t,n-t\}}} +\frac{2\sqrt{2\bar{a}^2m\log n}}{\sqrt{\min\{t,n-t\}}}\\
& \le d_i +\frac{6\sqrt{2\bar{a}^2m\log n}}{\sqrt{\min\{t,n-t\}}}
\end{align*}
where the second line comes from the feasibility of the scaled LP solved at time $t$. Due to the symmetry of the random permutation, 
\begin{align*}
    \E \left[a_{i,t+1}I(r_{t+1}>\Aa_{t+1}^\top \p_{t+1})\right] 
& \le d_i +\frac{6\sqrt{2\bar{a}^2m\log n}}{\sqrt{\min\{t,n-t\}}}.
\end{align*}
Summing up the inequality, we have 
$$\E[\bm{A}x-\bm{b}]\le O(\sqrt{mn}\log n).$$

\end{proof}

\subsection{Proof of Theorem \ref{theorem_PBD}}

\label{proofTheoPBD}
\begin{proof}
At time $t$, the optimal solution to the scaled LP is $\tilde{\bm{x}}^{(t)}=(\tilde{x}_1^{(t)},...,\tilde{x}_t^{(t)})$. We have
\begin{equation*}
    \begin{split}
        \mathbb{E}\left[r_tx_t\right]
        &= 
        \mathbb{E}\left[r_t\tilde{x}_t^{(t)}\right]\\
        &=
        \frac{1}{t}\mathbb{E}\left[\sum_{j=1}^tr_s \tilde{x}^{(t)}_j\right].
    \end{split}
\end{equation*}
Then, for the objective,
\begin{equation*}
    \begin{split}
        R_n^* - \sum_{t=1}^n \E\left[r_t x_t\right]
        & = 
        R_n^* - \sum\limits_{t=1}^{n} \frac{1}{t}\mathbb{E}\left[\sum_{j=1}^tr_s \tilde{x}^{(t)}_j\right] \\
        & \leq
        m\bar{r} 
        + 
        \frac{\bar{r}}{\underline{d}}\log n\sqrt{n}
        +
        m\bar{r}\log n
        +
        \bar{r}\max\{16\bar{a}^2,e^{16\bar{a}^2},e\}.
    \end{split}
\end{equation*}
where the second line comes from an application of Proposition \ref{importantLemma}. Then, we analyze the constraint violation. From the construction of the algorithm, we have that $\mathbb{E}[a_{it}x_t]\leq d_i$. 
Let 
$$A_{it} = a_{it}x_t - d_i$$
and then we know 
$$M_{it} = \sum_{j=n-t+1}^{n} A_{ij}$$ is a supermartingale with $|A_{ij}|\leq \bar{a}+\bar{d}$. Then if we apply Hoeffding's lemma for supermartingale, we have 
\begin{equation*}
    \begin{split}
        \mathbb{P}\left(
            M_{in}\geq
            2(\bar{a}+\bar{d})\sqrt{n}\log n
        \right)
        &\leq
        \exp\left\{-
        \frac{2(\bar{a}+\bar{d})^2n\log^2 n}{n(\bar{a}+\bar{d})^2}
        \right\}\\
        &\leq
        \exp\{-2\log^2 n\}
        \leq
        \frac{1}{n},
    \end{split}
\end{equation*}
when $n>3$. Thus, 
\begin{equation*}
    \begin{split}
        \mathbb{E}
        \left[
        \left(\sum\limits_{t=1}^{n}a_{it}x_t-d_i\right)^+
        \right]
        &=
        \mathbb{E}\left[(M_{in})^+
        \right]\\
        &\leq
        2(\bar{a}+\bar{d})\sqrt{n}\log n
        \mathbb{P}\left(
            M_{in}<
            2(\bar{a}+\bar{d})\sqrt{n}\log n
        \right)\\
        &\ \ \ \ +
        \bar{a}n
        \mathbb{P}\left(
            M_{in}\geq
            2(\bar{a}+\bar{d})\sqrt{n}\log n
        \right)\\
        &\leq
        2(\bar{a}+\bar{d})\sqrt{n}\log n+\bar{a}\\
        \mathbb{E}
        \left[v(\bm{x})
        \right]
        &\leq
        2(\bar{a}+\bar{d})\sqrt{mn}\log n
        +
        \bar{a}\sqrt{m}.
    \end{split}
\end{equation*}

\end{proof}

\subsection{Proof of Theorem \ref{theorem_Feasibel}}

\begin{proof}
For the upper bound of $v$, lemma \ref{iidBound} and
$$
\sum\limits_{t=1}^{n}a_{it}x_t-b_i
\leq
\sqrt{n}p_{n+1,i}
\leq
\sqrt{n}\|\bm{p}_{n+1}\|_{2},
$$
can give the result.

For the fesibility,
first, we prove that if the initial solution is infeasible for the $i$-th constraint, i.e., $\sum\limits_{t=1}^{n} a_{it}x_{t}>nd_i$, then 
\begin{equation}
    \sum\limits_{t\in S_0}a_{it} \geq v\sqrt{n}\log n
    \label{feasibleCond}
\end{equation}
with high probability. According to the definition of $v$, the corresponding $\hat{x}_t$ is feasible.
To start, (\ref{feasibleCond}) holds with probability 1 if $|S_0|=n_+$ which implies $S_0=S_+$. Otherwise, the infeasibility indicates $n_+\geq\frac{\underline{d}n}{\bar{a}}\geq 1$, and by definition $n_+ \le n.$

This implies bounds on the cardinality of $S_0$,
$$|S_0| \ge \frac{2v\sqrt{n}\log n}{\bar{a}}$$
and
$$|S_0| \le \frac{2v\sqrt{n}\log n}{\underline{d}}+1 \le \frac{3v\sqrt{n}\log n}{\underline{d}}.$$
Consequently,
$$\frac{|S_0|}{n_+}\sum\limits_{t\in S_+}a_{it}
        \geq
        \frac{2v\log n}{\underline{d}\sqrt{n}}n\underline{d}
        \geq
        2v\sqrt{n}\log n.$$ Then, since $n\geq\left(\frac{6\bar{a}}{\underline{d}}\right)^4$
\begin{equation*}
    \begin{split}
        \mathbb{P}
        \left(
        \sum\limits_{t\in S_0}a_{it}
        < v\sqrt{n}\log n
        \right)
        &\leq
        \mathbb{P}
        \left(
        \sum\limits_{t\in S_0} a_{it}
        -
        \frac{|S_0|}{n_+}
        \sum\limits_{t\in S_+}
        a_{it}
        \leq -v\sqrt{n}\log n
        \right) \\  
       &\leq
        \exp\left\{
            -\frac{2v^2n^2\log^2n}{4\bar{a}^2|S_0|^2}
        \right\}\\
        &\leq
        \exp\left\{
            -\frac{n\underline{d}^2}{18\bar{a}^2}
        \right\}   
        \leq \frac{1}{n^2}.
    \end{split}
\end{equation*}
Moreover, for any $i$ s.t. $\sum\limits_{t=1}^n a_{it} < \frac{n\underline{d}}{2}$,
since $n>16$ and $\sqrt{n}>\frac{12\bar{a}(\bar{r}+(\bar{a}+\underline{d})^2m)\log n}{\underline{d}^2}$, we have 
\begin{equation*}
    \sum\limits_{t\in S_+\backslash S_0}
    a_{it}
    \leq
    \frac{n\underline{d}}{2}
    +
    \bar{a}|S_0|
    \leq
    n\underline{d}.
\end{equation*}
Then, for any $i$ s.t. $\frac{n\underline{d}}{2}\leq\sum\limits_{t=1}^{n}a_{it}\leq{n\underline{d}}$. Similarly, we can find that
\begin{align*}
        \frac{|S_0|}{n_+} \sum\limits_{t\in S_+}a_{it}
        &\geq
        \frac{v\log n}{\underline{d}\sqrt{n}}n\underline{d}
        \geq
        v\sqrt{n}\log n,\\
        \mathbb{P}
        \left(
        \sum\limits_{t\in S_0}a_{it}
        <
        0
        \right)
        &\leq
        \mathbb{P}
        \left(
        \sum\limits_{t\in S_0} a_{it}
        -
        \frac{|S_0|}{n_+}
        \sum\limits_{t\in S_+}
        a_{it}
        \leq-v\sqrt{n}\log n
        \right) \\      
        &\leq
        \exp\left\{
            -\frac{2v^2n^2\log^2 n}{4\bar{a}^2|S_0|^2}
        \right\}\\
        &\leq
        \exp\left\{
            -\frac{n\underline{d}^2}{18\bar{a}^2}
        \right\}
        \leq
        \frac{1}{n^2}.
\end{align*}

Combining three parts above,  we have that with probability at least $1-\frac{2m}{n^2}\geq 1-\frac{2}{n}$, the modified solution is feasible. Given the fact that the modified solution change the original solution for at most $O(\sqrt{n}\log n)$ entries, the modified solution can achieve $O((m+\log n)\sqrt{n})$ regret.
\end{proof}

\end{document}